\documentclass[11pt,letterpaper]{article}
\usepackage[T1]{fontenc}
\usepackage{graphicx}

\usepackage{geometry}
\usepackage[utf8]{inputenc} % allow utf-8 input
\usepackage{hyperref}       % hyperlinks
\usepackage{url}            % simple URL typesetting
\usepackage{booktabs}       % professional-quality tables
\usepackage{amsfonts}       % blackboard math symbols
\usepackage{nicefrac}       % compact symbols for 1/2, etc.
\usepackage{microtype}      % microtypography
\usepackage{xcolor}         % colors
\usepackage{amsmath}
\usepackage{amsthm}
\usepackage{enumitem}
\usepackage{comment}

\usepackage{tikz}
\usetikzlibrary{shapes.geometric, positioning, arrows.meta, fit}
\usepackage{subcaption}

\newtheorem{alg}{Mechanism}
\usepackage{cleveref}

\newtheorem{definition}{Definition}
\newtheorem{lemma}{Lemma}
\newtheorem{remark}{Remark}
\newtheorem{theorem}{Theorem}
\newtheorem{proposition}{Proposition}
\newtheorem{corollary}{Corollary}
\newtheorem{observation}{Observation}

\title{Position Auctions with a Capacity Constraint}

\usepackage{titling}
\usepackage[noblocks]{authblk}

\author[1]{Eleni Batziou}
\affil[1]{Archimedes/Athena RC, Greece.}
\author[2]{Georgios Birmpas}
\affil[2]{School of Computer Science, University of Liverpool, UK.}
\author[2]{Georgios Chionas}
\author[3]{Piotr Krysta}
\affil[3]{School of Computer and Cyber Sciences, Augusta University, USA.}

\date{}

\begin{document}
\newcommand{\edp}{$\e$-differential privacy }
\newcommand{\eddp}{$(\e, \delta)$-differential privacy }
\newcommand{\edpm}{$\e$-differential private mechanism}
\newcommand{\e}{\epsilon}
\newcommand{\EXP}{\text{EXP}}
\newcommand{\an}{$\alpha n$}
\newcommand{\uout}{$u^{\text{out}}$}
\newcommand{\uinf}{$u^{\text{inf}}$}

\newcommand{\bfx}{\mathbf{x}}
\newcommand{\bfp}{\mathbf{p}}
\newcommand{\bfq}{\mathbf{q}}
\newcommand{\tfm}{(\allocs,\prices,\burn)}
\newcommand{\fpa}{(\allocs^f,\prices^f,\burns^f)}

\newcommand\fnsep{\textsuperscript{,}}

\newcommand{\bid}{b}
\newcommand{\bids}{{\mathbf \bid}}
\newcommand{\bidsmt}{{\mathbf \bid}_{-i}}
\newcommand{\bidt}[1][t]{{\bid_{#1}}}

\newcommand{\epnebid}{b^*}

\newcommand{\val}{v}
\newcommand{\vals}{{\mathbf \val}}
\newcommand{\valsmt}{{\mathbf \val}_{-i}}
\newcommand{\valt}[1][t]{{\val_{#1}}}

\newcommand{\alloc}{x}
\newcommand{\allocs}{{\mathbf \alloc}}
\newcommand{\allocsmt}{\allocs_{-i}}
\newcommand{\alloct}[1][t]{\alloc_{#1}}

\newcommand{\price}{p}
\newcommand{\prices}{{\mathbf \price}}
\newcommand{\pricet}[1][t]{\price_{#1}}

\newcommand{\burn}{q}
\newcommand{\burns}{{\mathbf \burn}}
\newcommand{\burnt}[1][t]{\burn_{#1}}

\newcommand{\cost}{\gamma}
\newcommand{\oca}{(\bids,\allocs,\bm{\tau})}

\newcommand{\nine}{*}
\newcommand{\ninetfm}{(\allocs^{\nine},\prices^{\nine},\burns^{\nine})}

\newcommand{\history}{\mathbf{H}}

\newcommand{\blocks}{\mathcal{B}}
\newcommand{\blockset}{\mathcal{B}}
\newcommand{\vbp}{v_{BP}}
\newcommand{\vbphat}{\hat{v}_{BP}}
\newcommand{\basefee}{r}

\newcommand{\Rplus}{\mathbb{R}_{+}}

\newcommand{\OPT}{\text{OPT}}
\newcommand{\SOL}{\text{SOL}}
\newcommand{\REV}{\text{REV}}
\newcommand{\SURPLUS}{\text{SURPLUS}}
\newcommand{\ALG}{\text{ALG}}
\newcommand{\Sgr}{S_{\text{greedy}}}

\newcommand{\al}{\alpha}

\newcommand{\xopt}{x^{\text{OPT}}}

\newcommand{\Wt}{W^{(t)}}

\newcommand{\vminus}{\mathbf{v}_{-i}}

\newcommand{\MAX}{\text{MAX}}

\newcommand{\spaces}{\mathbf{W}}

\newcommand{\allocation}{\mathbf{x}} % changed from A, there's also another notation under \allocs

\newcommand{\capratio}{k \cdot \frac{s_{max}}{W}}
\newcommand{\dec}{r}

\maketitle

\begin{abstract}
Sponsored search auctions are commonly modeled as an assignment of a fixed set of positions to a set of advertisers, with welfare maximization reducible to a standard matching problem. Motivated by modern advertising environments, we study a richer variant of the classical position auctions model, in which ads may have heterogeneous sizes and the platform must jointly select a subset of ads and assign them to positions subject to a global capacity constraint. The size of an ad can capture the physical space it occupies on a page or its consumption of a scarce platform resource, such as communication bandwidth or computational capacity. Thus, size measures the resource cost imposed by an ad on the platform, which need not be related to the value generated by displaying it. 

We begin by studying a general formulation of the matching with a capacity constraint problem and propose an algorithmic technique that goes beyond known greedy methods while achieving constant-factor approximation guarantees for social welfare. Our allocation rule augments density-based ordering with capacity-aware local improvements, which allow for re-allocations that improve welfare, while respecting the capacity constraint. Applied in the context of position auctions, we analyze this algorithm under the standard assumptions of single-parameter agents, publicly known ad sizes, and position-dependent click-through rates (CTRs). In this setting, the size of an ad determines the amount of capacity it consumes, while the CTR depends only on the position to which it is assigned. We show that a minor modification to our algorithmic approach yields a universally truthful mechanism that achieves a 4-approximation of the optimal social welfare, and we show that this guarantee is tight. To the best of our knowledge, this is the first such mechanism for this variant of the capacity-constrained matching.
\end{abstract}

%\keywords{Sponsored Search, Ad Auctions, Approximation Algorithms, Truthfulness}

%%%%%%%%%%%%%%%%%%%%%%%%%
\section{Introduction}\label{sec:int}
%%%%%%%%%%%%%%%%%%%%%%%%%

Online advertising platforms are responsible for allocating scarce resources, such as impressions and clicks, to advertisers seeking to maximize their utility. A canonical instance is that of \emph{sponsored search (keyword) auctions} where a user’s query triggers an auction, in which ads compete for assignment to positions (slots) on a result page. Search advertising generates billions in revenue globally, making the design of efficient and truthful mechanisms a key problem from an economic, as well as theoretical perspective. In the standard position auction model, introduced in the seminal papers of Varian \cite{Varian07Position} and Edelman et al. \cite{Edelman2007}, the search platform allocates a fixed set of slots to advertisers and charges payments based on submitted bids. Early models assume that each ad occupies a fixed amount of page space, leading to the standard abstraction that advertisers compete for inclusion and assignment to positions, where the realized value of an allocated agent is a product of her private value and the associated publicly-known click-through rate (CTR) of the given position. Determining the welfare-maximizing allocation reduces to a standard bipartite matching problem \cite{gale1962college, shapley1971assignment}.

While this abstraction has been fundamental in the sponsored search literature, it does not fully capture a variety of advertising environments used by modern platforms \cite{Cavallo17RichAds, Hartline2018fastRichAds, Ghiasi19RichAds, AggarwalBMM022}. Ads can differ in their formats, sizes, and requirements, and the platform may have to decide not only which ads receive priority, but also which ads can be included at the same time. For example, a platform may choose between a compact text ad, a larger product ad, or an ad containing an image or video, while ensuring that the selected ads fit within the available display capacity. The size of an ad need not refer only to the physical space that it occupies on a page. It can also represent the amount of a scarce platform resource that is required to display or serve the ad. Different ads may consume different amounts of communication bandwidth because of their format or content, require different computational resources for rendering, personalization, or processing, or impose different costs on the platform’s infrastructure when delivered to users. Similar constraints may arise in applications or environments with limited memory, energy, screen space, or processing capacity. In each of these settings, the platform has a limited capacity and the inclusion of an ad uses part of that capacity.

This interpretation of size is distinct from the value that an ad generates once it is included. An ad may require more bandwidth or computational resources without generating a different value when assigned to a particular position. Likewise, even when size represents physical display space, the platform may evaluate the expected value of an ad and the space that it consumes as separate attributes. Thus, the size of an ad can affect the feasibility of an allocation without affecting the value associated with assigning it to a position. This motivates a model in which ads have heterogeneous sizes and the total size of the selected ads must satisfy a global capacity constraint. In this richer setting, the platform faces a problem of jointly selecting a subset of ads that fits within the available capacity, and assigning the selected ads to positions in a way that maximizes total social welfare. 

Further applications to this variant of the problem, besides conventional web pages with heterogeneous ad formats, include settings in which sponsored content is embedded in dynamically generated environments. Inspired by advances in generative AI, a stream of recent work considers extended position auction settings where buyers compete for the placement of sponsored creatives within AI-generated commercial content, such as ads embedded within LLM-generated summaries \cite{Balseiro2025, Dubey2024}. In such settings, the available display environment may impose limits on the number, size, or resources of the ads that can be included. More generally, the same abstraction applies whenever a platform selects several items that compete both for priority and for a common, limited resource.

Considering this problem from an optimization viewpoint, instead of a simple bipartite matching, the winner determination problem must now be formulated as a special version of budgeted bipartite matching\footnote{In our case, we refer to the budget constraint as capacity constraint, due to the studied setting.}  \cite{Berger11BudgetedMatching}, which is known to be NP-hard via a reduction from the knapsack problem \cite{Berger11BudgetedMatching}. Given that our work examines the problem from a game-theoretic perspective, our goal is the design of mechanisms that are incentive-compatible and thus the central challenge that we address is the following:

\begin{quote}
 \emph{Can we design truthful mechanisms for the capacity-constrained matching, that are both computationally efficient, and provide strong guarantees on social welfare? }  
\end{quote}

Motivated by the described setting, we study the bipartite matching problem under a capacity constraint. Without considering incentives, we first propose an efficient algorithm that provides a constant-factor approximation to the social welfare for the general version of the problem. We then move to the position auctions setting, and study the problem from a game-theoretic perspective. In particular, we adopt the standard assumptions of single-parameter agents, publicly known ad sizes, and position-dependent click-through rates (CTRs). In this setting, the size of an ad determines the amount of capacity consumed by the allocation, while the CTR determines the expected number of clicks associated with the position to which the ad is assigned. The goal is to design mechanisms that are polynomial-time computable and achieve a good approximation to the social welfare, while being incentive-compatible, which ensures that agents cannot benefit from misreporting their true values. We show that a small modification to our original algorithm can turn it into a universally truthful mechanism that still achieves constant-factor approximation, and thus combines the two key properties in auction design.

%%%%%%%%%%%%%%%%%%%%%%%%%%%%%%%
\subsection{Contributions}\label{sec:con}
%%%%%%%%%%%%%%%%%%%%%%%%%%%%%%%

We study position auctions augmented with a global capacity constraint. Our motivation comes from online advertising auctions where participants compete both for inclusion and priority in a limited-capacity environment. We formulate the problem as a capacity-constrained matching between agents (that each has a single ad) and positions, where each agent has a single value for inclusion, and where agents may have different sizes. Our goal is to design simple truthful mechanisms that achieve good approximation guarantees with respect to the social welfare.
\begin{itemize}

\item In Section \ref{sec:warm}, we begin by proving that the problems that we consider are NP-hard. We then demonstrate that existing approaches that involve greedy techniques which consider the bids and densities of the agents, fail to achieve meaningful approximation guarantees. Therefore, we identify fundamental challenges in adapting known methods from simpler and more restricted constrained environments, such as the knapsack problem.
\item In Section \ref{sec:alg}, we propose a novel algorithm for the general version of capacity-constrained bipartite matching that yields a 6-approximation of the optimal social welfare. Our algorithm combines simple greedy selection with adaptive reassignment to correct allocation errors. We additionally show a lower bound of 3 on its approximation ratio. %by this algorithm. 
\item In Section \ref{sec:mech}, we study the problem from a mechanism design perspective. Agents have private values for being included in the solution and are, thus, single-parameter. Moreover, ad sizes are public and each position has a publicly known CTR, same for each agent, that captures the expected number of clicks that an ad placed there will receive. The number of positions is fixed and the total size of winning ads must respect the global capacity constraint. Note that, in the mechanism design setting, on one hand, we impose a stricter form on the agents' realized values, compared to the algorithmic setting, since now all agents share the same ranking over positions. On the other hand, we require the mechanism to be monotone. 
We first discuss why the algorithm introduced in \Cref{sec:alg} is not monotone (and thus cannot be directly transformed to a truthful mechanism). We then introduce a critical modification of the stopping condition, that renders this new version of the algorithm monotone, and therefore truthful, when coupled with the  appropriate payment rule \cite{Myerson}. By randomizing between this mechanism and the algorithm that outputs the single highest-value feasible pair, we obtain a universally truthful mechanism that achieves a 4-approximation of the optimal social welfare. We additionally show that this guarantee is tight, as there are instances in which the approximation ratio of the mechanism is arbitrarily close to $4$. Finally, we complement our theoretical results with simulations on synthetic data, which evaluate the average performance of our mechanisms.
\end{itemize}

To the best of our knowledge, this is the first work to study this variant of the capacity-constrained matching problem and to give constant approximation guarantees with a monotone algorithm. Our algorithm can be viewed as a variant of greedy-guided local search. Local search algorithms are notoriously non-monotone, highlighting the  technical challenges we have to overcome.  For more on the technical challenges and our approaches, we refer the reader to the respective sections.

%To the best of our knowledge, this is the first work to study this variant of the capacity-constrained matching problem, and give constant approximation guarantees with a monotone algorithm.
%In fact, our algorithm can be viewed as a variant of greedy-guided local search, that maintains a feasible solution and performs local improvements. Local search algorithms are notoriously non-monotone. In our case, performing a search that is governed by a global heuristic, instead of exhaustively exploring the neighborhood, is one of the key components that guarantee monotonicity. For more on the technical challenges and our approaches, we refer the reader to the respective sections. 

%%%%%%%%%%%%%%%%%%%%%%%%%%%%%%%%%
\subsection{Related work}\label{sec:rew}
%%%%%%%%%%%%%%%%%%%%%%%%%%%%%%%%%

\paragraph{Sponsored Search Auctions.}
 The works of Varian \cite{Varian07Position} and Edelman et al. \cite{Edelman2007} introduce position auctions and study the equilibria of Generalized Second Price auctions (GSP). The early literature on sponsored search auctions assumes that click-through rates (CTRs) are publicly known. This has initiated a long line of work on position auctions with applications in sponsored search \cite{Athey2011}.

 Subsequent work has generalized the GSP approach and considered its performance beyond the standard model  \cite{Ghosh07Conciseness,Cavallo14Independent}. This line of work studies position auctions with different CTRs across advertisers. Moreover, Deng et al. \cite{Deng10RichAds} were the first to study a \textit{rich ad problem} where each ad could occupy multiple slots, while Deng et al. \cite{DengGSTZ17} study the pricing of ad slots when each advertiser demands a block of consecutive slots. Cavallo et al. \cite{Cavallo17RichAds} also consider the rich ad problem combined with constraints on the number of ads, keeping the problem single-parameter since advertisers cannot misreport their set of rich ads. In fact, even though it is not explicitly specified, the underlying optimization problem of Cavallo et al. \cite{Cavallo17RichAds} is a version of the budgeted bipartite matching. They propose a local search algorithm to configure the allocated ads, without, however, providing any theoretical guarantees on its performance. Additionally, they consider truthfulness, under the assumption that agents aim to maximize their utility and are indifferent to the payment, provided it does not exceed their value. %Works that regard alternatives of truthfulness include \cite{AggarwalAM2006}, while there are also approaches that consider efficiency and equilibrium properties of non-truthful mechanisms. We refer the reader to \cite{lahaie2007sponsored, Qin21} for extensive surveys.

Closely related to our work, Aggarwal et al. \cite{AggarwalBMM022} propose truthful  mechanisms for settings with heterogeneous size ads under knapsack constraints. The underlying optimization problem in their setting is the Multiple-Choice Knapsack Problem, which is a special case of the matching with a capacity constraint that we consider. The main distinction from our setting is that, in their work, the set of positions that each advertiser is interested in does not intersect with those of others. Therefore, agents do not compete for the same positions.
Finally, we note that budget constraints have been studied in the context of position auctions \cite{Lu2015, Ashlagi2010}. These works propose auction schemes akin to the GSP in the presence of personalized budgets for each advertiser. In contrast, the focus of our paper is on an overall budget (capacity) constraint that restricts the number of displayed ads, %that can be displayed, 
rather than agent-specific limits. Budget constraints also appear in the literature on budget-feasible mechanisms \cite{Singer10,AmanatidisBM16}, where the budget bounds the total payments of the mechanism, rather than the total size of the allocated ads.

\paragraph{Knapsack Auctions.} 
Aggarwal and Hartline \cite{AggarwalH06} study the knapsack problem from a game-theoretic perspective, modeling a setting where agents with private values and publicly known sizes compete for space in a knapsack. Briest et al. \cite{BriestKV11} develop techniques for converting approximation algorithms into truthful mechanisms with bounded efficiency loss, applicable to knapsack allocation problems. 
% TODO (Piotr): check and complete the following sentence, explaining why the techniques of [BriestKV11] do not directly yield a monotone FPTAS for our problem.
In our setting, however, the allocation does not only determine whether an agent wins, but also the position she is assigned to, and monotonicity must hold across all positions. 
Dughmi and Roughgarden \cite{DughmiR14} show that every binary packing problem that admits an FPTAS also admits a truthful-in-expectation mechanism that is an FPTAS. Truthfulness in expectation is weaker than universal truthfulness, as it only guarantees that truthful bidding maximizes the expected utility of an agent over the randomness of the mechanism, while the mechanism that we propose is universally truthful and has a simple combinatorial description.
Moreover, Grandoni et al. \cite{Grandoni2014utilitarian} consider the budgeted matching problem and devise FPTAS mechanisms, though by resorting only to truthfulness with high probability, and allowing bounded violations of the budget constraint. Last but not least, a related line of work studies greedy allocation rules, in particular forward greedy algorithms, which are attractive due to their simplicity and approximate efficiency. However, D\"utting et al. \cite{DuttingGR17} show that such mechanisms, and as such the special class of deferred-acceptance auctions, face limitations in achieving strong approximation guarantees in knapsack settings.

%%%%%%%%%%%%%%%%%%%%%%%%%%%%%%%%%%%%%%%%%
\section{Preliminaries}\label{sec:prelem}
%%%%%%%%%%%%%%%%%%%%%%%%%%%%%%%%%%%%%%%%%

We study the bipartite matching problem with a capacity constraint. Our setting consists of a set $N = \{1, \dots, n\}$ of agents (advertisers), each owning a single ad and a set $K = \{1, \dots, k\}$ of available positions. In the algorithmic problem, each agent $i \in N$ has a value $v_{ij} \geq 0$ for having her ad placed at position $j \in K$. Each ad has a size $s_i > 0$, and the page imposes a hard capacity constraint $W > 0$, namely the total size of included ads cannot exceed $W$. \footnote{Our problem can be seen as a special case of budgeted bipartite matching problem presented in \cite{Berger11BudgetedMatching}, where ad sizes are independent of the positions they are assigned to. Specifically, in a general bipartite graph, both values and sizes can vary between each edge.} 

In the position auctions variant, we assume that each agent $i \in N$ holds a private value $v_i > 0$ for having her ad included in the page, while her ad size $s_i > 0$ is publicly known. The domain is therefore single-parameter. Each position $j \in K$ has a publicly known click-through rate (CTR) $a_{j}> 0$, capturing the  expected number of clicks that an ad placed there receives. We assume CTRs are decreasing: $a_1 > a_2 > \dots > a_k>0$. If an ad $i$ is allocated at position $j$, then the \emph{realized value} of agent $i$ is $v_{ij} = v_i \cdot a_j$. We further define the density of an agent-position pair as $d_{ij} = v_{ij}/ s_i$. The feasibility structure of our problem is a capacity-constrained \emph{matching} between agents and positions.
%; each agent holds a single ad, can be assigned to at most one position, and each position can be assigned to at most position, and the total size of the assigned agents must not exceed the total capacity constraint.
Note that the position auctions setting restricts the algorithmic problem, as it implies that all agents share the same preferences over the orderings of the positions.

Each agent $i$ submits a bid $b_i$ for inclusion in the page, and we denote the full bid vector by $\bids = (b_1, b_2, \dots, b_n)$ and the bids of all agents except $i$ by $\bids_{-i}$. 
% The goal is to design a mechanism that takes a vector of bids and determines a subset of winning ads to include in the page, along with an allocation of each ad to a specific position.
% \footnote{We will use the term $b_{-i}$ to refer to the vector $(b_1, b_2, \dots, b_{i-1}, b_{i+1}, \dots, b_n)$ of all bids except $b_i$.} 
A mechanism is defined as a tuple $(\allocs, \prices)$, consisting of an allocation rule $\allocs : \mathbb{R}^{n}_{\ge 0} \to \{0,1\}^{n \times k}$,  
% $\allocs : \mathbb{R}^{n \times k}_{\ge 0} \to \{0,1\}^{n \times k}$ 
where $x_{ij} = 1$ indicates that ad $i$ has been placed at position $j$, along with a payment rule $\prices : \mathbb{R}^n_{\ge 0} \to \mathbb{R}^n_{\ge 0}$ with $p_i(\bids)$ denoting the price charged to agent $i$. 
The utility of each agent $i$ is quasilinear and defined as 
$u_i (\bids) = \sum_{j \in K} \{a_{j} \cdot v_i \cdot x_{ij} (\bids) \} - p_i(\bids)$. We restrict our attention to \emph{individually rational} mechanisms, where any agent not assigned a position pays 0, while her utility is also defined as 0.

Our primary objective is to maximize social welfare, defined as the total realized value of allocated agents:
$SW(\allocs) = \sum_{i \in N} \sum_{j \in K} v_{ij} \cdot x_{ij}$. Any feasible allocation must satisfy the capacity constraint $\sum_{i \in N} \sum_{j \in K} s_i \cdot x_{ij} \le W$.
%\footnote{This is also our objective when we consider the more general algorithmic version of the problem described earlier}.
We aim at designing efficient mechanisms that are truthful and achieve strong approximation guarantees with respect to the optimal social welfare. We now introduce some relevant definitions.

\begin{definition}[(Universally) Truthful Mechanisms]
A deterministic mechanism $(\allocs, \prices)$ is truthful if for every agent $i$ with true value $v_i$ and every $\bids_{-i}$
it holds that $u_i(v_i, \bids_{-i}) \ge u_i(b_i, \bids_{-i})$ for all $b_i \neq v_i$. A randomized mechanism is universally truthful if it is a probability distribution over deterministic truthful mechanisms.    
\end{definition}

%Truthfulness ensures that no agent can benefit from misreporting her true value, regardless of the bids of the other agents. 

\begin{definition}[Monotone Allocation Rule]

%An allocation rule $\allocs$ is monotone if, for a bidder $i$ and every fixed bid profile $b_{-i}$ of agents excluding $i$, declaring $b'_i \geq b_i$ gives $u_i (b'_i, b_{-i}) \geq u_i (b_i, b_{-i})$.

An allocation rule $\allocs$ is monotone if for every agent $i$ and fixed values $\bids_{-i}$, increasing $b_i$ cannot result in agent $i$ being assigned to a position of strictly lower realized value. 

%with valuations $v_i \ge v'_i$, it results in $u_i(v_i, v_{-i}) \ge u_i (v'_i, v_{-i})$.
\end{definition}

Since our domain is single-parameter, truthfulness reduces to showing monotonicity of the allocation rule. By Myerson's lemma \cite{Myerson}, every monotone allocation rule admits a unique payment scheme that makes the overall mechanism truthful.

%%%%%%%%%%%%%%%%%%%%%%%%%%%%%%%%%%%%%%%%%%%%%%%%%%%%%%%%%%%
\section{Warm Up: Limitations of Existing Techniques}
\label{sec:warm}
%%%%%%%%%%%%%%%%%%%%%%%%%%%%%%%%%%%%%%%%%%%%%%%%%%%%%%%%%%%

As we have discussed in Section \ref{sec:prelem}, we focus on two problems: the matching with a capacity constraint, and the position auctions variant. The latter is a simplification of the former as it assumes that agents share the same preferences over the orderings of positions. Both of these problems can be viewed as complex extensions of the classic knapsack problem, and thus they are NP-hard. We give the proof for the position auctions variant (therefore hardness of the matching with a capacity constraint follows).

\begin{proposition}\label{prp:hardness}
The position auctions with a capacity constraint problem is NP-hard.
\end{proposition}

\begin{proof}
We will use knapsack for the reduction. Consider an arbitrary instance of knapsack, with $n$ agents, capacity $W$, positive integer values $v_i$ and sizes $s_i$, and let $V=\sum^n_{i=1} v_i$. We construct an instance of the position auctions with a capacity constraint problem, that is the same as the knapsack one, that additionally has $n$ positions, and for the CTRs we have $a_j=1+\frac{n+1-j}{n\cdot V+1}$. This transformation can clearly be done in polynomial time. Before we proceed, consider any assignment in the position auctions setting, and let $S$ to be the agents that are part of this assignment, and $\pi_i$ to be the position assigned to agent $i \in S$. Notice that the total value that this assignment produces is $$\sum_{i \in S} a_{\pi_i}\cdot v_i = \sum_{i \in S}  v_i+\sum_{i \in S} \frac{(n+1-\pi_i) \cdot v_i}{n\cdot V+1}\leq \sum_{i \in S}  v_i+ \frac{n \cdot \sum_{i \in S}v_i}{n\cdot V+1} \leq \sum_{i \in S}  v_i+ \frac{n \cdot V}{n\cdot V+1}<\sum_{i \in S}  v_i+1.$$ 

We are ready to conclude our proof. Let $V_k$ to be the optimal knapsack value produced by set $S_k$, and $V_p$ to be the optimal position auction with a capacity constraint value in the constructed instance, produced by an assignment where agents from set $S_p$ participate. First of all, note that $V_p \geq V_k$ as the optimal value in the position auction problem, is at least the value produced in this problem by the agents in $S_k$ in any assignment, and thus the best as well. This value is strictly better than $V_k$ (recall that the CTRs are greater than one, and that $S_k$ respects the capacity constraint). On the other hand, notice that if $V_k\geq \sum_{i\in S_p}v_i+1$, then we would have $V_p<\sum_{i\in S_p}v_i+1 \leq V_k$, a contradiction. Therefore, $V_k<  \sum_{i\in S_p}v_i+1$, and since the values are integers, we get that $V_k\leq \sum_{i\in S_p}v_i$. Finally, since $V_k$ is the optimal such sum, we get  $V_k= \sum_{i\in S_p}v_i$. Our proof is complete.
\end{proof}

Given the above result, and as our starting point,  we investigate the applicability of well-known greedy approaches for Knapsack, to our setting. We focus on the approximation guarantees of these methods and examine whether they continue to hold in our framework. 

As a remark, for all the mechanisms that we will consider in this paper, we assume that there is a pre-processing step that removes agents with sizes strictly larger than $W$. These agents do not affect either the approximation guarantees of an algorithm (as they cannot be part of a feasible solution), or the monotonicity of a mechanism (as sizes are public information and cannot be manipulated by the agents). Similarly, in Section \ref{sec:mech}, we assume that agents that submit a bid of $0$ are removed in the pre-processing step, and thus they are never allocated. This does not affect the approximation guarantees, as such agents contribute zero value to any solution, nor the monotonicity of a mechanism, as an agent that bids $0$ is not allocated, and thus she cannot be worse off by increasing her bid.

\paragraph{Greedy-by-Density.} Consider the monotone density ranking $r_{d}$ that orders densities $d_{ij} = \frac{v_{ij}}{s_i}$ decreasingly. We now define the greedy allocation rule $G(r_{d})$ as follows: greedily assign the agent-position pair $(i,j)$ with the highest density $d_{ij}$ as long as the capacity constraint is not violated and agent $i$ has not been previously allocated. Position $j$ is now considered allocated, and for every $i' \neq i$, the respective densities $d_{i'j}$ are removed from the ordering.

\paragraph{Greedy-by-Value.} Consider the monotone value ranking $r_{v}$ that orders values $v_{ij}$ decreasingly. We now define the greedy allocation rule $G(r_{v})$ as follows: greedily assign the agent-position pair $(i,j)$ with the highest value $v_{ij}$ as long as the capacity constraint is not violated and agent $i$ has not been previously allocated. Position $j$ is now considered allocated, and for every $i' \neq i$, the respective $v_{i'j}$ are removed from the ordering.

\begin{definition}
 The $\max(G(r_{d}),G(r_v))$ allocation rule runs the Greedy-by-Density and Greedy-by-Value subroutines and outputs the allocation with the highest social welfare between the two.
\end{definition}

This is a natural modification of the max-greedy approach that provides a 2-approximation for the knapsack problem. However, as the following proposition demonstrates, this approximation guarantee does not hold even if one considers that more restricted position auctions variant of the problem.

\begin{proposition}\label{prp:bad_instance}
Consider the position auctions with a capacity constraint setting. The allocation rule $\max(G(r_{d}),G(r_v))$ achieves a $k$-approximation of the optimal social welfare, where $k$ is the number of positions. Moreover, this bound is tight.
\end{proposition}

\begin{proof}
The upper bound comes from the simple observation that the output of $G(r_v)$ always includes the feasible agent-position pair $(i,j)$ with the highest value $v_{ij}$ subject to $s_i \leq W$, thus the welfare of $G(r_v)$ is at least as high as $v_{max} = \max_{i,j: s_i \leq W} v_{ij}$. Since any feasible allocation is a matching assigning at most one agent per position and an instance with $k$ positions can have at most $k$ agents assigned, the optimal welfare can be at most equal to $k \cdot v_{max}$. Thus, the welfare of $\max(G(r_{d}),G(r_v))$ is at least a $k$-approximation of the optimal solution.

Regarding the lower bound, consider an instance with $k $ positions, $2k + 1$ agents, and assume a total capacity of $W =k$. Also, let the CTRs be as follows: $a_j=1+\epsilon_j$, where if $j<j'$, then $\epsilon_j>\epsilon_{j'}$, and for all $j$ we have that $\epsilon_j>0$ is arbitrarily small. We partition the agents into three sets as follows:
\begin{itemize}
    \item $S_0$: one single heavy agent with $s_0 = W = k$ and $v_0 = 1 + \epsilon$, for some arbitrarily small $0<\epsilon< \frac{1}{2 \cdot \sum^k_{j=1} (1+\epsilon_j)}$
    \item $S_1$: $k$ agents, each with $s_i = 1$ and $v_i = 1$
    \item $S_2$: $k$ agents, each with $s_i = \epsilon$ and $v_i = 2 \epsilon$
\end{itemize}
Given that, the densities satisfy
$d_{S_2} (k) = 2 a_k > d_{S_1} (1) = a_1 > d_{S_0} (1) = a_1 (1+\epsilon)/k$. Ties are resolved lexicographically. We first examine the welfare produced by $\max(G(r_{d}),G(r_v))$. Initially, notice that $G(r_{v})$ assigns the agent in $S_0$ to position 1, with welfare $(1+\epsilon_1) \cdot (1+\epsilon)= 1+\epsilon^*$. Then, it halts since $s_0 = W$. At the same time, $G(r_{d})$ assigns the agents of $S_2$ to positions $1,\dots, k$ since they have the highest densities. Then, it halts as they cover all $k$ positions. The total welfare achieved by $G(r_d)$ is $2 \epsilon \cdot \sum^k_{j=1} (1+\epsilon_j)<1$. Therefore, the total produced welfare is the maximum between the two portions, $1+\epsilon^*$. On the other hand, the optimal solution assigns the agents in $S_1$ to positions $1, \dots, k$. This results to a total welfare of $\sum^k_{j=1} (1+\epsilon_j)$, for which we have $k<\sum^k_{j=1} (1+\epsilon_j)< k\cdot (1+\epsilon_1)< k\cdot (1+\epsilon^*)$. Our proof is complete. 
\end{proof}

This leads us to the next section, where we present a novel approach to the general version of the problem, which (as we show) overcomes the inefficiencies of existing techniques.

%%%%%%%%%%%%%%%%%%%%%%%%%%%%%%%%%%%%%%%
\section{Matching with a Capacity Constraint: A Constant-Approximation Algorithm}\label{sec:alg}
%%%%%%%%%%%%%%%%%%%%%%%%%%%%%%%%%%%%%%%

The optimization problem underlying this setting is a knapsack, augmented with matching constraints. A feasible solution must respect the global capacity constraint, as well as being a valid matching between agents and positions: each agent can be allocated to at most one position, and each position must be assigned to at most one agent. As shown in \Cref{sec:warm}, this additional complexity causes natural greedy heuristics, such as density or value-based greedy, to perform poorly in our setting, resulting in linear (in terms of the number of available positions) approximations to social welfare.

The approach that we introduce goes beyond such greedy methods. We propose an algorithm that processes agent-position pairs in decreasing order of density, but treats assignments as tentative rather than permanent. This means that an agent assigned to a position based on her density during some step of the mechanism, might later be displaced by an agent that has higher value for that position, provided that the latter fits in the current solution. This adaptive reassignment step enables the mechanism to correct poor greedy choices. Finally, we randomize between the output of this algorithm and an algorithm that outputs the single highest-value feasible pair, coined $G(v_{max})$, and show that this yields a 6-approximation for the capacity-constrained matching. 

Below, we describe our mechanism in detail:

\begin{alg}
\label{alg:augm_greedy}
Consider a value profile $(\vals_1, \vals_2, \dots , \vals_n)$, where $\vals_i = (v_{i1}, v_{i2}, \dots , v_{ik})$.
Let $w$ denote the available space at any stage of the algorithm, initialized to $w=W$. Sort all agent-position pairs $(i,j)$ in decreasing order of density $d_{ij}= v_{ij} / s_i$. Process pairs in this order as follows:
% Following this ordering, proceed as follows: Consider the current $d_{ij}= v_{ij} / s_i$,
\begin{enumerate}
    \item If position $j$ is currently empty, then if $w-s_i<0$, stop.
    %proceed to the next element in the ordering.
    Otherwise, assign agent $i$ to position $j$ and set $w \leftarrow w-s_i$. Remove all remaining pairs $(i,j')$ associated with agent $i$ from the density ordering and continue with the next highest ranked element in the list.
    %Remove any density $\frac{v_{ij'}}{s_i}$ that is associated with agent $i$ from the density ordering. Continue to the next highest ranked element in the updated ordering.
    \item If position $j$ is currently occupied by some agent $i'$, then if  $w+ s_{i'}-s_{i}<0$, stop. 
    Otherwise:
    \begin{itemize}
        \item  If $v_{ij}>v_{i'j}$, or if $v_{ij}=v_{i'j}$ and $s_i<s_{i'}$, or if $v_{ij}=v_{i'j}$ and $s_i=s_{i'}$, and $i$ precedes $i'$ in the density ordering, replace agent $i'$ with $i$, and update $w \leftarrow w+s_{i'}-s_i$. Remove all remaining pairs $(i,j')$ from the density ordering, and return all pairs $(i',j'')$ associated with agent $i'$ except those for positions that $i'$ has previously been assigned to and examined for\footnote{Meaning the densities of this agent that the mechanism has already processed.}. Continue from the highest ranked element in the updated ordering.
        \item Else, move to the next highest ranked element in the ordering.
    \end{itemize}
\end{enumerate}
\textbf{Output}: The produced allocation $\allocation$.
\end{alg}
%\el{Consider putting the stopping condition as a while loop at the start instead of within each step.}

The ties between agents are broken lexicographically. If an agent has the same density for two or more positions, ties are again broken lexicographically. We note that once agent $i$ is examined for position $j$, being assigned to it or not (based on steps 1 and 2 of Mechanism \ref{alg:augm_greedy}), the pair $(i,j)$ is never revisited and the agent is never examined for position $j$ again. Therefore, Mechanism \ref{alg:augm_greedy} terminates after at most $n\cdot k$ steps. 

We proceed by showing that Mechanism \ref{alg:augm_greedy} alone is not enough to guarantee a constant-factor approximation.

\begin{proposition}\label{prp:bad_instance_main_result_1}
Mechanism \ref{alg:augm_greedy} provides an arbitrarily bad approximation of the maximum social welfare.
\end{proposition}
\begin{proof}
    Consider an instance with 2 agents and 2 positions and some arbitrarily small $\epsilon>0$. Let $v_{11}=1+\epsilon$, $v_{12}=\epsilon$ and $s_1=1$, while $v_{21}=\epsilon$, $v_{22}=W$ and $s_2=W$. Based on the ordered density list, Mechanism \ref{alg:augm_greedy} first puts agent 1 to position 1, and then examines agent 2 for position 2 and stops as agent 2 does not fit in the solution. This allocation provides a social welfare equal to $1+\epsilon$, while the optimal solution achieves social welfare $W$, where $W$ can be arbitrarily large.
\end{proof}

This observation motivates the following combination\footnote{Given that both options are assigned equal probabilities $\frac{1}{2}$, this can equivalently be transformed to a deterministic mechanism with the same approximation guarantees, that always selects the highest-welfare between the two options.}:

\begin{alg}
\label{alg:final_mech}
Consider a value profile $(\vals_1, \vals_2, \dots , \vals_n)$, where $\vals_i = (v_{i1},v_{i2}, \dots , v_{ik})$. Let $\allocation_1$ be the output of Mechanism \ref{alg:augm_greedy} and $\allocation_2$ the output of $G(v_{max})$.
\textbf{Output}: With probability $\frac{1}{2}$ output $\allocation_1$ and probability $\frac{1}{2}$ output $\allocation_2$.
\end{alg}

% Notice that the reason we randomize between the two parts is because the example of \Cref{prp:non-mon-knap} holds for the deterministic version of this mechanism as well (if we try to combine the two parts by using the $\max$ operator).

%%%%%%%%%%%%%%%%%%%%%%%%%%%%%%%%%%%%%%%%%%%%%%%%%
\subsection{Approximation Guarantees}\label{sec:aprx_guar_main_resul}
%%%%%%%%%%%%%%%%%%%%%%%%%%%%%%%%%%%%%%%%%%%%%%%%%%

In this section, we analyze the approximation guarantees of Mechanism \ref{alg:final_mech}. We first introduce some necessary notation that will be used throughout the proof. Our algorithm constructs, step by step, a matching (each agent is being assigned to at most one position) subject to the capacity constraint. We start by denoting that $(i,j) \in M$, if agent $i$ is assigned to position $j$ in a matching $M$. Let $W^G\leq W$ denote the total capacity used by Mechanism \ref{alg:augm_greedy} and let $M^G$ be the matching it produces. Let $M^*_G$ denote the optimal matching with available capacity $W^G$ and $M^*$ denote the optimal matching with available capacity $W$, covering capacity $W^* \leq W$. For any matching $M$ we write $v(M)=\sum_{(i,j) \in M} v_{ij}$ to define the social welfare achieved by matching $M$. 

The main idea behind our approach is to first try to bound the social welfare produced by matching $M^*_G$, which is optimal under a (possibly) more constrained space $W^G\leq W$. To this end, we identify a partition of the agents that are part of the optimal solution, that will enable us to bound the value of each of them in each set, by the value of an agent that is part of $M^G$. This will be done in a way so that each such value is used at most twice. A crucial notion that we will need for this, is the one of \emph{consideration} for an agent. In particular, we say that an agent has been \emph{considered} by Mechanism \ref{alg:augm_greedy}, if at some point she was assigned to a position (either temporarily, or permanently). The proof of Lemma \ref{lem:main_apprx} formally describes this process. We then use it as a stepping stone in order to derive the bound for the desired matching $M^*$ in Theorem \ref{thm:main_approx}. 

Before we begin, we give the next lemma that will be useful in both the proofs of Lemma \ref{lem:main_apprx} and Theorem \ref{thm:main_approx}.

\begin{lemma} [Density Dominance]
\label{lem:rearlemma} Let $A$, $B$ be two disjoint sets of elements, each element represented by non-negative value-size tuples $(v_i, s_i)$ and  density $d_i = \frac{v_i}{s_i}\geq 0$. If $d_j \geq d_k$ for all $j \in A$, $k \in B$ (density dominance) and $\sum\limits_{j \in A}s_j \ge \sum\limits_{k \in B} s_k$ (size dominance), then $\sum\limits_{j \in A}v_j \ge \sum\limits_{k \in B} v_k$.
\end{lemma}

\begin{proof}
Let $m$ be the smallest density in $A$, that is, $m = \min\limits_{j \in A}d_j$. Hence, by density dominance, for every $k \in B$, it holds $d_k \leq m$. We now upper-bound the total value of $B$.
\begin{equation}
\label{eq:chebyshev-upper-bound}
    \sum\limits_{k \in B}v_k = \sum\limits_{k \in B}s_k \cdot d_k \le \left(\max\limits_{k \in B} d_k  \right) \cdot \sum\limits_{k \in B}s_k \leq m \cdot \sum\limits_{k \in B}s_k.
\end{equation}
Similarly, we now lower-bound the total value of $A$. Since for every $j \in A$, it holds $d_j \ge m$, we get
\begin{equation}
\label{eq:chebyshev-lower-bound}
    \sum\limits_{j \in A}v_j = \sum\limits_{j \in A}s_j \cdot d_j \ge m \cdot \sum\limits_{j \in A} s_j.
\end{equation}
By combining inequalities \eqref{eq:chebyshev-upper-bound}, \eqref{eq:chebyshev-lower-bound} we get:
\begin{equation*}
    \sum\limits_{k \in B}v_k \leq m \cdot \sum\limits_{k \in B}s_k \le m \cdot  \sum\limits_{j \in A} s_j \le \sum\limits_{j \in A}v_j
\end{equation*}
where the second inequality holds due to the size dominance of set $A$ over $B$. 
\end{proof}

%For completeness, we provide the proof in the \Cref{app:pr-rearlemma}

We first state an observation on the density list of Mechanism \ref{alg:augm_greedy}, that will be useful in the proof of Lemma \ref{lem:main_apprx}.

\begin{observation}\label{obs:one-agent}\normalfont
Consider any point of the run of Mechanism \ref{alg:augm_greedy}, and let $m$ be the smallest density among the pairs that have been examined so far ($m=+\infty$ if no pair has been examined). Then there is at most one agent that has pairs in the current density list with density higher than $m$, and if such an agent exists, she is currently unassigned.
\end{observation}

\begin{proof}
We use induction on the number of examined pairs. Initially, no pair has been examined and the statement trivially holds. Suppose that it holds before the highest ranked pair $(i,j)$ of the list is examined. If $d_{ij}\leq m$, then, since $(i,j)$ is the highest ranked pair, every pair in the list has density at most $d_{ij}$, which becomes the new value of $m$. Therefore, after the examination, the only pairs that can have density higher than $m$ are pairs that are returned to the list, which belong to a single agent, namely the agent replaced by agent $i$ (if any), who is unassigned. If $d_{ij}>m$, then $m$ does not change, and by the induction hypothesis, agent $i$ is the unique agent that has pairs of density higher than $m$. If agent $i$ is assigned (to an empty position, or by replacing another agent), her pairs are removed from the list, and the only agent that may have pairs of density higher than $m$ is the agent replaced by agent $i$ (if any). Otherwise, the list only loses pair $(i,j)$, and the statement continues to hold. Finally, notice that an agent that has pairs in the list is unassigned, since the pairs of an agent are removed from the list when she is assigned.
\end{proof}

We can now show the following:

\begin{lemma}\label{lem:main_apprx}
   Consider matchings $M^*_G$ and $M^G$. It holds that $v(M^*_G) \leq 2\cdot v(M^G)+v(G(v_{max}))$.  
\end{lemma}

\begin{proof}
    Let $N^*_G$ be the agents that are included in the optimal solution given $W^G$ available capacity, and $S$ the set of agents that are included in the solution of Mechanism \ref{alg:augm_greedy}, given $W$ available capacity. We split set $N^*_G$ into subsets, and we bound their contribution separately. As we mentioned earlier, throughout the proof we use the term \emph{considered} in the following manner: We say that an agent has been considered by Mechanism \ref{alg:augm_greedy} if at some point it was assigned to a position. Notice that this agent, in the end, might be or might not be in the solution that Mechanism \ref{alg:augm_greedy} produces, but being assigned somewhere at least once (even if in the process she was replaced) is enough to be included in this set. 
    
    So, let $A \subseteq N^*_G $ be the set that contains every agent that is not in $S$, i.e., they are not included in the solution of Mechanism \ref{alg:augm_greedy}, and for which we have that the density for the position that has been assigned to them in $M^*_G$ is smaller or equal to the density of any pair $(i,j) \in M^G$. Notice that this set might contain agents that are both considered or not considered by Mechanism \ref{alg:augm_greedy}.
    Let $B \subseteq N^*_G $ be the set of agents that have not been considered by Mechanism \ref{alg:augm_greedy} (and thus, they are not in $S$) and their densities for the position that has been assigned to them in $M^*_G$ is higher than the density of at least one pair $(i,j) \in M^G$. 
    Moreover, let $C \subseteq N^*_G $ be the set of agents that have been considered by Mechanism \ref{alg:augm_greedy}, their densities for the position that has been assigned to them in $M^*_G$ is higher than the density of at least one pair $(i,j) \in M^G$, but are not in set $S$. Finally, let  $D \subseteq N^*_G $ be the set of agents that have been considered by Mechanism \ref{alg:augm_greedy} but are in set $S$, i.e., they are included in the solution of Mechanism \ref{alg:augm_greedy}. Notice that the above collection of sets defines a partition of set $N^*_G $.

    \textbf{Bounding the value of set $A$:} First of all, notice that the space that agents in $N^*_G \setminus D$ cover, i.e., their total size, is smaller or equal to the space that agents in $S\setminus D$ cover. The reason for this is that the solution of Mechanism \ref{alg:augm_greedy} covers exactly $W^G$ space while the optimal solution in this space might cover at most this (and notice that we remove set $D$ of common agents between the two). Therefore, the same goes for agents in $A\subseteq N^*_G \setminus D$. Moreover, agents in set $A$ have their densities for positions that have been assigned to them in $M^*_G$ to be smaller or equal to the density of any pair $(i,j) \in M^G$, and thus, of any pair in the subset of agents that belongs to $S\setminus D$. Therefore, by using Lemma \ref{lem:rearlemma}, the total contribution in terms of value that agents in set $A$ have in the optimal solution for space $W^G$ is smaller or equal to the total contribution that the agents in $S\setminus D$ have in the solution of Mechanism \ref{alg:augm_greedy}. So far we have used the agents in $S\setminus D$ once.

    \textbf{Bounding the value of set $B$:} Consider an agent $i \in B$ and let $j$ be the position that is assigned to $M^*_G$. Initially, we show that position $j$ is occupied by some agent $i'$ in $M^G$ as well. Suppose for contradiction that this is not the case. Now notice that the first scenario is that  Mechanism 1 reached at some point $d_{ij}$ and did not add $(i,j)$ in the solution. Observe that at that point agent $i$ is available as she has not been considered at all, by definition of set $B$, it must be the case that she did not fit in the solution because of the capacity constraint. Since Mechanism \ref{alg:augm_greedy} stops when it reaches a pair that cannot be included due to exceeding total capacity and no more agents are included after, we get that $d_{ij}$ is smaller or equal to any density of any pair in $M^G$\footnote{Notice that since $i\in B$, she could not be previously assigned to a position, and then be reinserted in the list of available densities, which is the only scenario under which $d_{ij}$ could be greater or equal to a density of some pair in $M^G$.}, a contradiction. The second scenario is that Mechanism 1 never reached $d_{ij}$, and stopped at some density before that. Therefore, once more, all the included pairs have higher densities than $d_{ij}$, a contradiction. So, we have that $(i',j)$ is in $M^G$. It is not hard to see that since $d_{ij}$ is higher than at least the density of one pair in $M^G$, the only way that she was not considered by Mechanism \ref{alg:augm_greedy} is that position $j$ was already occupied, Mechanism \ref{alg:augm_greedy} tried to assign agent $i$ to position $j$, agent $i$ could fit in the solution (as otherwise the mechanism would stop there and all the included densities would be higher than hers, a contradiction), but the value of the currently assigned agent to position $j$ was higher or equal to $v_{ij}$. The latter implies that $v_{ij}\leq v_{i'j}$. Notice that this is the first time that position $j$ is examined. For the assigned agent we have that if $i' \in S\setminus D $ we have used the contribution of pair $(i',j)$ twice so far, else if $i' \in D $ we have used the contribution of pair $(i',j)$ just once so far.

    \textbf{Bounding the value of set $C$:}  Consider an agent $i \in C$ and let $j$ be the position that is assigned to in $M^*_G$.  There are two possible scenarios: 1) Agent $i$ is not the agent of the pair at which Mechanism \ref{alg:augm_greedy} stops. First of all, notice that since agent $i$ has been considered by Mechanism \ref{alg:augm_greedy}, and  $d_{ij}$ is higher than at least the density of one of the elements in $M^G$, Mechanism \ref{alg:augm_greedy} will reach $d_{ij}$ before it terminates. To see this, suppose that this is not the case. Since agent $i$ is not in $S$, the pair $(i,j)$ is then in the density list when the mechanism terminates, as it was returned to the list the last time that agent $i$ was replaced. If the mechanism examines every pair without stopping, this is a contradiction. Otherwise, since $d_{ij}$ is higher than the density of a pair in $M^G$, which has been examined, by \Cref{obs:one-agent} agent $i$ is the unique agent that has pairs in the list of density higher than every examined density. The pair at which the mechanism stops is the highest ranked pair of the list, so its density is at least $d_{ij}$, and thus it also belongs to agent $i$, a contradiction. We will prove that position $j$ is occupied by some agent $i'$ in $M^G$. Suppose, for contradiction, that this is not the case. Now notice that, based on the previous argument, Mechanism 1 reached $d_{ij}$ but did not add $(i,j)$ to the solution. This could only occur if she did not fit in the solution because of the capacity constraint, which contradicts the case we examine. So, we get that $(i',j)$ is in $M^G$. As $(i,j)$ is not in $M^G$ and $d_{ij}$ was reached, this implies that either agent $i$ was considered for position $j$ and then replaced by someone of higher or equal value, or the occupier at that point of position $j$ when $d_{ij}$ was reached had higher or equal value. In both cases, we derive that $v_{ij}\leq v_{i'j}$. As this is the first time position $j$ is examined, agent $i'$ has been used twice so far if she also belongs to set $S\setminus D$, or once so far if she is in $D$. 2) Agent $i$ is the agent of the pair $(i^e,j^e)$ at which Mechanism \ref{alg:augm_greedy} stops. Notice that $j^e$ might be different from $j$, in which case the mechanism might stop before reaching $d_{ij}$. Unlike the analysis of the previous case, in this scenario position $j$ might not be occupied in $M^G$, so it is not possible to use a similar argument. However, notice that there can be only one agent in $C$ that satisfies this condition throughout the run of Mechanism \ref{alg:augm_greedy}, as the mechanism terminates once. In this scenario, we  bound $v_{ij}$ by using $v(G(v_{max}))$, since $s_i\leq W$.
    
    \textbf{Bounding the value of set $D$:} Consider an agent $i \in D$ and let $j$ be the position that is assigned to in $M^*_G$. We have two subcases: 1) Agent $i$ is assigned to a position $j'$ in $M^G$ for which $v_{ij}>v_{ij'}$. It is easy to see that position $j$ is occupied by some agent $i'$ in $M^G$. The reason is that $v_{ij}>v_{ij'}$ implies $d_{ij}>d_{ij'}$. Therefore, Mechanism \ref{alg:augm_greedy} would have considered position $j$ before position $j'$ for agent $i$. So now, assume that an agent $i'$ occupies position $j$ in $M^G$. Notice that it must hold that $v_{i'j}\geq v_{ij}$ as the only possible cases are that agent $i$ was considered for position $j$ and then replaced, or she was not considered at all, both implying a lower or equal value. As this is the first time position $j$ is examined, agent $i'$ has been used twice so far if she also belongs in $S\setminus D$, or once so far if she is in $D$. 2) Agent $i$ is assigned to a position $j'$ under $M^G$ for which $v_{ij}\leq v_{ij'}$. Now suppose that $v_{ij}< v_{ij'}$. Notice that position $j'$ is occupied by some agent $i'$ in $M^*_G$. The reason for this is that if this was not the case, putting agent $i$ there would provide a value that is higher than $v(M^*_G)$, a contradiction. In this case, agent $i$ will be used to bound the value of herself in $M^*_G$ and has been used once more depending on if agent $i'$ belongs in set $D$ (see subcase 1) or in sets $B, C$, thus at most twice in total. Else, if $v_{ij}= v_{ij'}$, then position $j'$ might or might not be occupied by some agent $i'$ in $M^*_G$. If it is, the same arguments as before apply, if it is not (or if $j=j'$), then agent $i$ will be used to just bound the value of herself (so at most once in this case). 
    \end{proof}

Now let $SW^R$ be the expected social welfare of Mechanism \ref{alg:final_mech}. Given the previous lemma, we can now prove the main theorem of this section.

\begin{theorem}\label{thm:main_approx}
   Mechanism \ref{alg:final_mech} achieves (in expectation) a 6-approximation of the optimal social welfare. Moreover, its approximation cannot always be better than 3.
\end{theorem}
\begin{proof}
    Assume that $W^* > W^G$. Consider space $W^G$ and the agents that are present in $M^*$. We build a new matching $M^t$ as follows: Suppose that $M^t$ is initially empty. First include in $M^t$ every pair $(i,j) \in M^*$ that involves agents that are also in $S$. It is easy to see that so far, capacity $W^G$ has not been exceeded. Now start putting inside the remaining pairs $(i,j)$ from $M^*$, in an arbitrary way, %in terms of their density
    until capacity $W^G$ is reached but not exceeded (i.e., if any of the remaining agents is included, then the sum of the sizes will be more than $W^G$).

    Now notice that Mechanism \ref{alg:augm_greedy} stops adding elements in the solution when it reaches a pair, say $(i^e, j^e)$, that has density $d_{i^e j^e}$ which is the highest among the remaining pairs and exceeds the current available capacity\footnote{If such a pair does not exist, Mechanism \ref{alg:augm_greedy} examines every pair without stopping. In this case, every pair $(i,j)\in M^*$ whose agent $i$ is not in $S$ was examined (or agent $i$ was assigned to position $j$ and later replaced), and since the mechanism never stops, position $j$ is occupied in $M^G$ by an agent of value at least $v_{ij}$. Charging these pairs to the occupants of their positions, and bounding the pairs of $M^*$ whose agents are in $S$ by $v(M^*_G)$ (their total size is at most $W^G$), we get $v(M^*)\leq v(M^*_G)+v(M^G)\leq 3\cdot v(M^G)+v(G(v_{max}))$, which is stronger than the bound derived below.}. Going back to the current matching $M^t$, consider set $M^*\setminus M^t$ and notice that it might contain pairs of  higher density than $d_{i^e j^e}$, and pairs of smaller (or equal) density than $d_{i^e j^e}$,% since pairs from $M^*$ are added to $M^t$ in an arbitrary order.
    
    Regarding the former, we have that since each one of the pairs (that have higher density than $d_{i^e j^e}$), say $(i',j')$, involve agents that do not occupy a position in $M^G$, they were reached by Mechanism \ref{alg:augm_greedy} and each one of them either was (eventually) replaced by the current occupier of position $j'$ because she had higher or equal value, or they were never considered for their specific position $j'$ because the then occupier had higher or equal value (and this cannot decrease through the run of the mechanism). Notice that it cannot be the case that they did not fit in the solution at the point they were reached, as the mechanism would have stopped there and not in $d_{i^e j^e}$.  Therefore, for each such position $j'$, there is a unique agent $i$ that is part of $M^G$ (the current occupier of this position) that contributes value more or equal to $i'$ in this position. Thus, we can bound the sum of the values of these agents by $v(M^G)$.

    Regarding the latter, notice that the sum of sizes of the agents that are part of $M^*\setminus M^t$, cannot be more than $2 \cdot s_{i^e}$ unless the cardinality of $M^*\setminus M^t$ is 1. To see this notice that $W-W^G<s_{i^e}$ (regardless of if position $j^e$ was occupied at the time the density $d_{i^ej^e}$ was reached in the ordering), and at the same time any agent from $M^*\setminus M^t$ that is added to $M^t$ exceeds $W^G$. So a sum of sizes (of 2 or more agents) that is more than $2 \cdot s_{i^e}$ would mean that since $W^*-W^G\leq W-W^G<s_{i^e}$, the remaining space between $W^G$ and what $M^t$ covers must be bigger than $s_{i^e}$. This would imply that all the agents in $M^*\setminus M^t$ (that belong to the case that we examine) would need to be bigger than $s_{i^e}$, as by our construction, none of them fits in what remains between $W^G$ and the space occupied by $M^t$. This leads to a contradiction, as adding any agent to $M^t$ creates an occupied space that is larger than $W^G$, and the available space that remains is smaller than $W^*-W^G\leq W-W^G<s_{i^e}$, so none of the other agents could fit in the optimal solution. Therefore, we conclude that either the sum of sizes of the agents that are part of $M^*\setminus M^t$, are less or equal to $2 \cdot s_{i^e}$, or the cardinality of $M^*\setminus M^t$ is 1. In the first case, their values are bounded by 2 times the value of $(i^e, j^e)$ because of Lemma \ref{lem:rearlemma}, while in case the cardinality of set  $M^*\setminus M^t$ is 1, the value of the agent that remains is bounded by $v(G(v_{max}))$. 
    
    Finally, in case $W^* \leq W^G$, then $M^*=M^t$. Putting everything together, 
    
     \begin{align*}
        v(M^*) & \leq v(M^t)+2\cdot v(G(v_{max}))+ v(M^G)\\
        &\leq  v(M^*_G) + 2\cdot v(G(v_{max}))+v(M^G)\\
        & \leq 3\cdot v(M^G)+  3\cdot v(G(v_{max}))\\
        &= 6\cdot (\frac{v(M^G)}{2}+\frac{v(G(v_{max}))}{2})\\
        & = 6\cdot SW^R, 
    \end{align*}
    where the first and second inequalities hold because $M_G^*$ is the optimal matching when the given space is $W^G$, and $v(G(v_{max}))$ the highest valued pair that fits in $W$ (therefore, higher than $v_{i^e j^e}$). The third inequality holds because of Lemma \ref{lem:main_apprx}, and finally the last equation comes from the definition of the expected social welfare (recall that Mechanism \ref{alg:final_mech} randomizes between $v(M^G)$ and $v(G(v_{max}))$ with probability $\frac{1}{2}$). 
    
    Regarding the lower bound, we provide a simple instance under which Mechanism \ref{alg:final_mech} obtains a 3-approximation of the optimal social welfare. Consider the following instance with 9 agents, 7 positions and space $W=6$. Also let $\epsilon$ to be an arbitrarily small positive amount. Agents 1 and 2 have a non-zero value only for the first position, with $v_{11} = 3, v_{21} = 3+2\epsilon$, and sizes $s_1 = \delta_1< \frac{ \epsilon}{3 + \epsilon}$ and $s_2 = 3+\delta_2< 3+\frac{3 \epsilon}{3 + \epsilon}$. Agent 3 has a non-zero value only for the second position $v_{32} = 3+\epsilon$ and size $s_3 = 3$. Therefore, it is not hard to verify $d_{21}>d_{32}$. Moreover, each agent with index $i = \{4, \dots, 8\}$ has value $v_{ij} = 1$ for all positions $j$ and size $s_i = 1$. Finally, agent 9 has value $v_{9j} = 1$ for all positions $j$ and size $s_9 = 1 - \delta_1$. It is not hard to verify that her density for all positions smaller than $d_{32}$.
    
    Mechanism \ref{alg:augm_greedy} will first assign agent 1 to position 1, and subsequently replace her by allocating agent 2 to position 1. Afterwards, it will examine agent's 3 density for position 2, but $s_2 + s_3 > W$, and the mechanism halts. The expected welfare obtained by Mechanism \ref{alg:final_mech} is $3+ 2 \epsilon$ (as both Mechanism \ref{alg:augm_greedy} and $G(v_{max})$ yield the same value), while the optimal solution would allocate agent 1 to position 1, agent 3 to position 2, and agent 9 along with two more agents from 4 through 8 to the remaining positions, achieving welfare of $9+\epsilon$. 
\end{proof}

%%%%%%%%%%%%%%%%%%%%%%%%%%%%%%%%%%%%%%%
\section{From Matching to Position Auctions with a Capacity Constraint}\label{sec:mech}
%%%%%%%%%%%%%%%%%%%%%%%%%%%%%%%%%%%%%%%

We now turn our attention to the position auctions setting from a mechanism design perspective. This is a special version of the general problem studied in \Cref{sec:alg}, as the \emph{realized} values of the agents are the product of their private value and the publicly known CTRs, which implies that all agents have the same ordinal preferences over positions. We now have to account for the strategic behavior of agents who may misreport to improve their outcome, which introduces an additional level of complexity. Shifting our goal from achieving good approximation alone, we are interested in designing mechanisms that are computationally efficient, achieve constant-factor approximation to the optimal social welfare as well as being truthful, which in single-parameter settings is equivalent to designing monotone algorithms. 

 A first natural question is whether Mechanism \ref{alg:augm_greedy} is also monotone. Unfortunately, the answer is negative. As shown in the proposition below, there are instances which demonstrate that Mechanism \ref{alg:augm_greedy} violates monotonicity. This, in turn, implies that Mechanism \ref{alg:final_mech} is not universally truthful. 

\begin{proposition}\label{prop:nonmon}
Mechanism \ref{alg:augm_greedy} is not monotone.
\end{proposition}
\begin{proof}
While the reallocation step helps achieve a constant-factor approximation, the following example illustrates that the mechanism is not monotone. Consider the following instance with 5 agents, 4 positions and capacity $W=6$. Each agent $i$ is defined by the tuple $(v_i, s_i)$. The agents are $(v_1, s_1)= (5,1)$, $(v_2,s_2)=(4+2 \epsilon,1)$, $(v_3,s_3)=(3, 1)$, $(v_4, s_4)=(4, 2)$ and $(v_5, s_5)=(5.5, 4)$. The CTRs are set as $a_1=1, \, a_2=1- \epsilon, \, a_3=1-2\epsilon, \, a_4=0.5$. 
Mechanism \ref{alg:augm_greedy} will first assign agent 1 to position 1, then agent 2 to position 2, and agent 3 to position 3. Subsequently, it will replace agent 3 with agent 4 at position 3, since $v_3 \cdot a_3  < v_4 \cdot a_3$ and eventually agent 3 will be assigned to position 4. Suppose now that agent 3 increases her bid to $v_3'=4+ \epsilon$. Under this bidding profile, the mechanism will proceed as follows. First, it will allocate agent 1 to position 1, then agent 2 to position 2, then agent 3 to position 3. Agent 3 can no longer be replaced by agent 4 in position 3, since, under the augmented bid, $v_3' \cdot a_3 > v_4 \cdot a_3$. Given the remaining available capacity, the algorithm will examine agent 5. Now agent 5 can replace agent 1 in position 1, since $v_5 \cdot a_1 > v_1 \cdot a_1$, triggering a chain of replacements in the positions that follow. In the next steps, agent 1 replaces agent 2 at position 2 and then agent 2 replaces agent 3 at position 3. At this point, the whole capacity has been consumed, resulting in agent 3 being unassigned. Hence, agent 3 ends up worse off despite increasing her bid, showing that monotonicity does not hold for the mechanism. 
\end{proof}

As the proof demonstrates, the main problem has to do with the existence of re-allocations that may be different between the two instances, along with the fact that this also causes a change to the available space. In other words, the agent that increases her bid might temporarily be removed from the current allocation, at a different point of time compared to the instance where her bid is not increased. At the same time, the available space might also be different next time she will be examined again by the mechanism. Therefore, although increasing her bid gives an agent a better place in the density ordering, and makes it more difficult for her to be replaced by another agent, the way in which the replacements and re-insertions are performed might eventually lead to her removal from the allocation. 

As we have discussed in our introduction, this is a common problem that techniques that are based on local searches (where re-assignments happen based on the local optima) have, and that is why trying to achieve monotonicity with this type of rules is usually very difficult. In the section that follows, we show that imposing a more sophisticated stopping condition to Mechanism \ref{alg:augm_greedy}, can render the resulting mechanism monotone, enabling us to derive a universally truthful mechanism with constant approximation, through randomization and by applying Myerson's payment scheme. 

\subsection{A Monotone Mechanism}\label{sec:monmech}

Intuitively, the issue with Mechanism~\ref{alg:augm_greedy} is that when an advertiser is being replaced, it is not guaranteed that she will be re-allocated. When there is a chain of replacements (caused by the original agent that was replaced), each replacement in a chain is checked against the capacity separately, and is performed immediately, so the chain may exhaust the capacity half-way through, as happens to agent~$3$ in the proof of Proposition~\ref{prop:nonmon}.
Notice also that the decision to stop cannot be based only on the size of the agent that is
currently examined, as the available space $w$ at that point depends on which replacements
took place before, and these in turn depend on the bids, including the bid of the agent
herself. The modification that we propose evaluates the whole chain at once: the mechanism
stops only when the allocation that would \emph{actually} result from examining the current
pair, including all re-allocations it causes, exceeds the capacity. In that case, none of
these re-allocations is performed.

\paragraph{Preliminaries.}
Throughout this section we consider the position auctions setting, where $v_{ij}=v_i\cdot
a_j$ and $a_1>a_2>\dots>a_k$. Recall that $a_k>0$. As discussed in Section~\ref{sec:warm}, we also assume that
$s_i\le W$ and $v_i>0$ for every agent~$i$, since agents with larger sizes or zero bids are removed in the pre-processing step.

Since CTRs are strictly decreasing, every agent $i$ satisfies $d_{i1}>d_{i2}>\dots>d_{ik}$.
Hence in the density ordering, pair $(i,j)$ always precedes the pair $(i,j+1)$.

Besides the density ordering, we use the \emph{value ordering} of the agents. It ranks the
agents in decreasing order of $v_i$; ties are broken in favor of the smaller size, and then
lexicographically. This is exactly the order used by the replacement rule of Mechanism~\ref{alg:augm_greedy}:
agent~$i$ replaces agent~$i'$ in position~$j$ if and only if $i$ precedes $i'$ in the value
ordering, that is, if $v_{ij}>v_{i'j}$, or $v_{ij}=v_{i'j}$ and $s_i<s_{i'}$, or
$v_{ij}=v_{i'j}$, $s_i=s_{i'}$ and $i$ precedes $i'$ lexicographically.

At any point of the run of the mechanism below, we say that agent~$i$ is \emph{eligible}
for position~$j$ if the pair $(i,j)$ has already been examined. Intuitively, agent~$i$
becomes eligible for position~$j$ once the mechanism reaches density $d_{ij}$. By the
remark above, an agent
who is eligible for position~$j$ is also eligible for every position $j'<j$.

\begin{alg}
\label{alg:augm_greedy_fix}
Consider a value profile $(v_1,v_2,\dots,v_n)$. Let $w$ denote the available space at any
stage of the algorithm, initialized to $w=W$, and start from the empty allocation. Sort all
agent-position pairs $(i,j)$ in decreasing order of density $d_{ij}=v_{ij}/s_i$, with ties
broken lexicographically. Process the pairs in this order. When pair $(i,j)$ is
examined, agent~$i$ becomes eligible for position~$j$, and:
\begin{enumerate}
\item If agent~$i$ is currently assigned to some position, or position~$j$ is occupied by
an agent that precedes $i$ in the value ordering, move to the next pair.
\item Otherwise, consider the \emph{tentative} allocation obtained from the current one as
follows. Agent~$i$ is assigned to position~$j$. If position~$j$ was occupied by some
agent~$i_1$, then $i_1$ moves to position $j+1$ if she is eligible for it, and becomes
unassigned otherwise. If $i_1$ moves to position $j+1$ and this position was occupied by
some agent~$i_2$, then $i_2$ moves to position $j+2$ if she is eligible for it, and becomes
unassigned otherwise; and so on, until an agent moves to an empty position or becomes
unassigned. Let $s^-$ be the size of the agent that becomes unassigned ($s^-=0$ if there is
none).
\begin{itemize}
\item If $w-s_i+s^-<0$, stop and output the current allocation.
\item Otherwise, replace the current allocation by the tentative one, update
$w\gets w-s_i+s^-$, and move to the next pair.
\end{itemize}
\end{enumerate}
\textbf{Output}: The current allocation, once all pairs have been examined or the mechanism stops.
\end{alg}

\begin{remark}\normalfont
The replacement steps are those of Mechanism~\ref{alg:augm_greedy}. As we show below
(Lemma~\ref{lem:scan} and Corollary~\ref{cor:monotone-size}), the agents in the current
allocation always occupy positions $1,2,\dots$ in the value ordering,
so a displaced agent always precedes the agent occupying the next position, and only her
eligibility needs to be checked (since, when she is displaced, her pair for the next position may not have been examined yet).

The difference from Mechanism~\ref{alg:augm_greedy} lies in the stopping test. First, the
capacity is checked only in step~2, that is, only when the examined pair would actually
change the allocation; a pair that satisfies the condition of step~1 never stops the
mechanism. In Mechanism~\ref{alg:augm_greedy}, instead, the capacity is checked before the
values are compared, so the mechanism may stop at a pair that would not change the
allocation. Second, Mechanism~\ref{alg:augm_greedy} checks each replacement of a chain
separately and commits it, so the chain may exceed capacity half-way through, leaving
the last displaced agent unassigned. Mechanism~\ref{alg:augm_greedy_fix} checks the net
effect $w-s_i+s^-$ of the whole chain, and performs none of it if the capacity would be
exceeded. 

Recall that an agent who becomes unassigned in step~2 is not lost: when her next pair is examined, she is treated as any other unassigned agent and may re-enter the allocation. Regarding running time, each chain moves at most $k$ agents, so Mechanism~\ref{alg:augm_greedy_fix} runs in time $O(nk\log(nk)+nk^2)$.
\end{remark}

In order to highlight the differences between Mechanisms \ref{alg:augm_greedy} and \ref{alg:augm_greedy_fix}, we give the following example that demonstrates how Mechanism \ref{alg:augm_greedy_fix} runs in the instances of Proposition~\ref{prop:nonmon}.

\paragraph{Example.}
Consider the instance of Proposition~\ref{prop:nonmon}. With $v_3=3$, Mechanism~\ref{alg:augm_greedy_fix} assigns
agent~$1$ to position~$1$ when $(1,1)$ is examined, agent~$2$ to position~$2$ at $(2,2)$, and
agent~$3$ to position~$3$ at $(3,3)$. All other pairs examined up to that point satisfy the condition of
step~1, since either the agent is already assigned or the position is occupied by an agent
that precedes her in the value ordering. When $(4,3)$ is examined, agent~$4$ replaces agent~$3$ in position~$3$. The pair $(3,4)$ has
not been examined yet, so agent~$3$ is not eligible for position~$4$ and becomes unassigned,
and $w=6-1-1-2=2$. When $(3,4)$ is examined, agent~$3$ is assigned to position~$4$ and $w=1$. Next, pair $(5,1)$ is examined. Agent~$5$ would take position~$1$, agent~$1$ would move to
position~$2$, agent~$2$ to position~$3$, and agent~$4$, who is not eligible for
position~$4$, would become unassigned. Since $w-s_5+s_4=1-4+2<0$, the mechanism stops without performing any of these changes, and outputs the current allocation, in which agent~$3$ is in position~$4$.

With $v'_3=4+\epsilon$, agent~$3$ is assigned to position~$3$ at $(3,3)$, and the pair
$(4,3)$ now satisfies the condition of step~1, since agent~$3$ precedes agent~$4$ in the value ordering. When $(5,1)$ is examined,
the initiated chain would eventually move agent~$3$ to position~$4$, for which she is now eligible, and $w-s_5=3-4<0$. Therefore the mechanism stops without performing any of these changes, and agent~$3$ stays in position~$3$. In contrast,
Mechanism~\ref{alg:augm_greedy} leaves her unassigned in this case.

\paragraph{The allocation after each pair.} Mechanism~\ref{alg:augm_greedy_fix} builds its allocation step by step, through chains of replacements. For the analysis, it is convenient to have an explicit description of the allocation that the mechanism holds after the first $t$ pairs have been examined, which depends only on the set of these $t$ pairs (see Definition~\ref{def:xP} and Corollary~\ref{cor:monotone-size}).
%For the analysis, it is convenient to describe the allocation of Mechanism~\ref{alg:augm_greedy_fix} after each
%examined pair directly, rather than through the sequence of chains.

Mechanism~\ref{alg:augm_greedy_fix} uses two orderings: the density ordering, which determines the order in which the pairs are examined, and the value ordering, which determines which of two agents precedes the other in steps~1 and~2. In the analysis we also need orderings that differ from these only in the way ties are broken. Given a value profile, we say that an ordering $\sigma$ of the agent-position pairs and an ordering $\pi$ of the agents are \emph{consistent} with the profile, when $\sigma$ lists the pairs in non-increasing order of density and $\pi$ lists the agents in non-increasing order of value, with ties in either ordering broken arbitrarily. Thus $\sigma$ and $\pi$ differ from the density and value orderings only in the way ties are broken, and they coincide with them when there are no ties. Moreover, since $d_{i1}>\dots>d_{ik}$, any such $\sigma$ lists the pairs of each agent~$i$ in the order $(i,1),(i,2),\dots,(i,k)$.

For a consistent pair of orderings $(\sigma,\pi)$, the \emph{run of Mechanism~\ref{alg:augm_greedy_fix} with $(\sigma,\pi)$} is the same procedure as Mechanism~\ref{alg:augm_greedy_fix}, where the pairs are examined in the order $\sigma$, and all comparisons between agents use $\pi$ instead of the value ordering. In particular, when $\sigma$ and $\pi$ are the density and value orderings of Mechanism~\ref{alg:augm_greedy_fix}, the run with $(\sigma,\pi)$ is exactly Mechanism~\ref{alg:augm_greedy_fix}. All definitions and results up to Lemma~\ref{lem:levels} hold for the run with any consistent pair $(\sigma,\pi)$, since their proofs only use that $\sigma$ lists the pairs in non-increasing density and $\pi$ lists the agents in non-increasing value. This is needed in the proof of monotonicity, where we compare runs with orderings that break ties in different ways.

 \paragraph{Roadmap.} The main difficulty in proving monotonicity is that the allocation of Mechanism~\ref{alg:augm_greedy_fix} is built through chains of replacements, and the stopping point depends on the whole sequence of these chains. When agent~$i$ increases her bid, her pairs move earlier in the density ordering and she may overtake other agents in the value ordering. This changes which chains are performed, and may also change the point where the mechanism stops. In particular, the run may continue further than before, and the danger is that the additional pairs examined bring agents of higher value that push agent~$i$ to a lower position, or out of the allocation. This is exactly what breaks the monotonicity of Mechanism~\ref{alg:augm_greedy}. 

We begin by giving a direct description of the allocation that the mechanism holds after the first $t$ pairs have been examined, which depends only on the set of these pairs. This allows us to compare two runs by comparing the allocations they induce on the same set of pairs, without following the chains that produced them.

\begin{definition}\label{def:pairs}
Fix a consistent pair $(\sigma,\pi)$. Let $P_t$ denote the set of the first $t$ pairs of
$\sigma$, for $t=0,1,\dots,n\cdot k$. Given $P=P_t$, agent~$i$ is \emph{eligible} for
position~$j$ if $(i,j)\in P$, and $\ell_i(P)$ denotes the largest position for which
agent~$i$ is eligible ($\ell_i(P)=0$ if there is none). Agent~$i$ is eligible exactly for
positions $1,\dots,\ell_i(P)$.
\end{definition}

\begin{definition}\label{def:xP}
Given $P=P_t$, the allocation $x(P)$ is defined as follows. For $j=1,2,\dots,k$ in turn,
position~$j$ is assigned to the agent that comes first in $\pi$ among the agents that are
eligible for~$j$ and have not been assigned to positions $1,\dots,j-1$. If there is no such
agent, positions $j,\dots,k$ remain empty.
\end{definition}

In other words, each position is assigned to the highest-value agent among those whose
density for it has already been examined and who were not allocated to an earlier position.
We denote by $|x|$ the total size of the agents allocated in $x$. 

\begin{lemma}\label{lem:scan}
The allocation $x(P)$ can be obtained as follows. Examine the agents in the order $\pi$,
keeping a counter $c$ that is initially $0$. When agent~$i$ is examined, if
$\ell_i(P)\ge c+1$, assign her to position $c+1$ and increase $c$ by one; otherwise she
remains unassigned.
\end{lemma}

\begin{proof}
Suppose that positions $1,\dots,j-1$ have been assigned, so that $c=j-1$. Every agent that
has been examined so far and was not assigned was rejected at some counter value $c'\le j-1$.
This means that $\ell_{i}(P)<c'+1\le j$, so she is not eligible for position~$j$.
Therefore, the next agent to be assigned is the first agent in $\pi$ that is not yet
assigned and is eligible for position~$j$, which is exactly the agent that
Definition~\ref{def:xP} assigns to position~$j$.

If there is no such agent, then no later position can be filled either, since eligibility
for position $j'>j$ implies eligibility for position~$j$.
\end{proof}

Consequently, the occupied positions in $x(P)$ are always positions $1,\dots,m$ for some
$m\le k$, and they are occupied by agents in the order~$\pi$. For an agent~$i$, let $c_i(P)$ denote the value of the counter just before agent~$i$ is
examined in Lemma~\ref{lem:scan}. Then agent~$i$ is assigned in $x(P)$ if and only if
$\ell_i(P)\ge c_i(P)+1$, and in that case she is assigned to position $c_i(P)+1$. Notice that
$c_i(P)$ depends only on the agents that precede $i$ in $\pi$, and not on $\ell_i(P)$.

Moving forward, we can now explore how the allocation changes when a single pair is examined. We show that examining a pair causes at most one chain, and that the total size and the social welfare of the allocation never decrease along the run.

\begin{lemma}\label{lem:scanfacts}
Consider the value scan of Lemma~\ref{lem:scan}, applied to a sequence of agents with given
eligibilities $\ell_g$.
\begin{enumerate}[label=(\alph*)]
\item If the eligibility of some agents increases, then the value of the counter at every
point of the scan weakly increases.
\item Consider scanning the same sequence of agents twice, once starting from counter $c$ and
once starting from counter $c+1$. Then either the same agents are assigned in both scans, or
the scan starting from $c+1$ assigns the same agents except for exactly one agent~$g$.
In the latter case, agent~$g$ is the first agent that the scan starting from $c$ assigns
to a position equal to her eligibility $\ell_g$. The agents examined before~$g$ that are
assigned in both scans occupy, in the scan starting from $c+1$, the position immediately
below the one they occupy in the scan starting from~$c$. The agents examined after~$g$ are
assigned to the same positions in both scans.
\end{enumerate}
\end{lemma}

\begin{proof}
(a) By induction on the scan. Suppose that the two counters before an agent~$g$ satisfy
$c\le\tilde c$, and that her eligibility satisfies $\ell_g\le\tilde\ell_g$. If $g$ is not
assigned in the first scan, the first counter does not change, while the second cannot
decrease. If $g$ is assigned in the first scan and $\tilde c=c$, then
$\tilde\ell_g\ge\ell_g\ge c+1=\tilde c+1$, so she is assigned in the second scan as well.
If $g$ is assigned in the first scan and $\tilde c>c$, then $\tilde c\ge c+1$ already. In
all cases, the inequality between the counters is preserved.

(b) As long as the two counters differ by one, every agent assigned in the scan with the
larger counter (that is, $\ell_g\ge c'+2$, where $c'$ is the smaller counter) is also
assigned in the scan with the smaller counter, one position higher. The first agent on
which the two scans disagree must therefore be assigned only in the scan with the smaller
counter, which means $\ell_g=c'+1$, that is, she is assigned to the position equal to her
eligibility $\ell_g$. After this agent the two counters coincide, and the two scans proceed
identically.
\end{proof}

\begin{lemma}\label{lem:one-pair}
Let $(i',j')$ be the pair that follows $P$ in $\sigma$, and let $P'=P\cup\{(i',j')\}$.
Suppose that the current allocation is $x(P)$ when $(i',j')$ is examined. Then:
\begin{itemize}
\item the pair $(i',j')$ satisfies the condition of step~1 of Mechanism~\ref{alg:augm_greedy_fix} if and only if $x(P')=x(P)$;
\item otherwise, the tentative allocation of step~2 is exactly $x(P')$. Agent~$i'$ is
assigned to position~$j'$, and the agents of $x(P)$ in positions $j',\dots,j''-1$ move one
position lower. Either no agent becomes unassigned, in which case $j''-1$ is the last
occupied position of $x(P)$, or exactly one agent~$i''$, who occupied position
$j''=\ell_{i''}(P)\ge j'$ in $x(P)$, becomes unassigned.
\end{itemize}
Moreover:
\begin{enumerate}[label=(\roman*)]
\item if the pair is handled by step~2, the size $s^-$ of the agent that becomes unassigned satisfies $s^-\le s_{i'}$ (with $s^-=0$ if no agent becomes unassigned), so $|x(P')|=|x(P)|+s_{i'}-s^-\ge|x(P)|$;
\item the value of the agent occupying each position weakly increases, so
$SW(x(P'))\ge SW(x(P))$;
\item $SW(x(P'))-SW(x(P))\le v_{i'j'}$;
\item at every point of the scan of Lemma~\ref{lem:scan}, the counter under $P'$ is at least
the counter under $P$ and at most one larger.
\end{enumerate}
\end{lemma}

\begin{proof}
Only the eligibility of agent~$i'$ changes, from $j'-1$ to $j'$. Let $c=c_{i'}(P)$, and
recall that positions $1,\dots,c$ of $x(P)$ are occupied by agents that precede $i'$ in
$\pi$, and the later positions by agents that follow her.
\begin{itemize}
\item If $c<j'-1$, agent~$i'$ is assigned in $x(P)$, since $\ell_{i'}(P)=j'-1\ge c+1$, and
also in $x(P')$; the scan does not change, and $x(P')=x(P)$. The pair satisfies the condition of step~1.
\item If $c\ge j'$, agent~$i'$ is unassigned in both $x(P)$ and $x(P')$, since
$\ell_{i'}(P')=j'\le c$, and $x(P')=x(P)$. Position~$j'\le c$ is occupied by an agent that
precedes $i'$, so the pair satisfies the condition of step~1.
\item If $c=j'-1$, agent~$i'$ is unassigned in $x(P)$, and position~$j'$ is either empty or
occupied by an agent that follows $i'$; the pair is handled by step~2. In $x(P')$, agent~$i'$ is
assigned to position~$j'$, and the agents after her in $\pi$ are scanned starting from
counter $j'$ instead of $j'-1$. By Lemma~\ref{lem:scanfacts}(b), the agents of $x(P)$ in
positions $j',j'+1,\dots$ move one position lower, up to the first of them, $i''$, whose
position $j''$ equals $\ell_{i''}(P)$, that is, who is not eligible for position $j''+1$;
she becomes unassigned, and all other agents keep their positions. Since $j''\ge j'$ and
$x(P)\ne x(P')$, this is exactly the chain of step~2.
\end{itemize}

(i) If no agent becomes unassigned, then $s^-=0$ and the claim is immediate. Otherwise, $s^-=s_{i''}$, and the pair $(i'',j'')$ belongs to $P$, so it precedes $(i',j')$ in $\sigma$, and
$d_{i''j''}\ge d_{i'j'}$. Moreover $v_{i''}\le v_{i'}$, because $i''$ comes after $i'$ in
$\pi$, and $a_{j''}\le a_{j'}$. Therefore
\[
s_{i''}=\frac{v_{i''}a_{j''}}{d_{i''j''}}\le\frac{v_{i'}a_{j'}}{d_{i'j'}}=s_{i'}.
\]

(ii) Position~$j'$ passes from its previous occupant, if any, to $i'$, who precedes that
occupant in $\pi$. Each position $p+1$, for $j'\le p<j''$, receives the previous occupant
of position~$p$, who precedes the previous occupant of position $p+1$. All other positions
are unchanged.

(iii) Writing $v_{(p)}$ for the value of the previous occupant of position~$p$, the change in
social welfare is
\[
a_{j'}v_{i'}+\sum_{p=j'}^{j''-1}a_{p+1}v_{(p)}-\sum_{p=j'}^{j''}a_pv_{(p)}\;\le\;a_{j'}v_{i'},
\]
since $a_{p+1}\le a_p$. If no agent becomes unassigned, the same computation holds without
the last term.

(iv) In the first two cases the two scans are identical. In the third case, the counters
coincide up to agent~$i'$; from agent~$i'$ until agent~$i''$ (or until the end of the scan,
if no agent becomes unassigned) the counter under $P'$ is larger by exactly one; and after
agent~$i''$ the counters coincide again, by Lemma~\ref{lem:scanfacts}(b).
\end{proof}

Consequently, the mechanism outputs the allocation at the last point of the density ordering where the allocation still fits in the capacity. 

\begin{corollary}\label{cor:monotone-size}
As long as it does not stop, Mechanism~\ref{alg:augm_greedy_fix} holds the allocation $x(P_t)$ after examining $t$
pairs, and it stops at the first $t$ with $|x(P_t)|>W$. The total size $|x(P_t)|$ and the
social welfare $SW(x(P_t))$ are non-decreasing in~$t$. Therefore Mechanism~\ref{alg:augm_greedy_fix} outputs
$x(P_T)$, where $T$ is the largest index with $|x(P_T)|\le W$.
\end{corollary}

\begin{proof}
The first claim follows by induction on $t$ from Lemma~\ref{lem:one-pair}, since the
tentative allocation of step~2 has size $|x(P_t)|=|x(P_{t-1})|+s_i-s^-$, so the test
$w-s_i+s^-<0$ holds exactly when $|x(P_t)|>W$. The second claim
follows from Lemma~\ref{lem:one-pair}(i) and~(ii), and the third from the first two.
\end{proof}

We now establish how the position of agent~$i$ reacts to changes in eligibility. For an agent~$i$, let $x^{i}(P)$ denote the allocation of Lemma~\ref{lem:scan} computed
with $\ell_i(P)$ replaced by $k$; that is, agent~$i$ is treated as eligible for every
position. Notice that if agent~$i$ is assigned in $x(P)$, then $x(P)=x^{i}(P)$, since the
eligibility of agent~$i$ only matters for whether she is assigned.

\begin{lemma}\label{lem:forced}
If agent~$i$ is assigned in $x^{i}(P_T)$, then $|x^{i}(P_t)|\le|x^{i}(P_T)|$ for every
$t\le T$.\footnote{The hypothesis is necessary: without it, the total size of $x^{i}(P_t)$
may decrease when agent~$i$ is pushed below position~$k$, since her size is not bounded as
in Lemma~\ref{lem:one-pair}(i).}
\end{lemma}

\begin{proof}
By Lemma~\ref{lem:scanfacts}(a), the counter $c_i(P_t)$ is non-decreasing in $t$. Since it
is at most $k-1$ at $T$, it is at most $k-1$ for every $t\le T$, so agent~$i$ is never
pushed below position~$k$.

Examining a pair of agent~$i$ does not change $x^{i}(P)$. For a pair of any other agent, the
argument of Lemma~\ref{lem:one-pair} applies, and the agent~$i''$ that becomes unassigned is
not agent~$i$, because $i$ could only become unassigned by being pushed below
position~$k$. Hence $(i'',j'')\in P$ is an actual examined pair, and $s_{i''}\le s_{i'}$ as
before.
\end{proof}

\begin{lemma}\label{lem:levels}
Increasing $\ell_i(P)$ cannot move agent~$i$ to a lower position. Increasing $\ell_h(P)$
for an agent $h\ne i$ cannot move agent~$i$ to a higher position.
\end{lemma}

\begin{proof}
Agent~$i$ is assigned if and only if $\ell_i(P)\ge c_i(P)+1$, and in that case she is
assigned to position $c_i(P)+1$. Since $c_i(P)$ does not depend on $\ell_i(P)$, the first
claim follows. For the second claim, $c_i(P)$ weakly increases with $\ell_h(P)$ by
Lemma~\ref{lem:scanfacts}(a), while $\ell_i(P)$ does not change.
\end{proof}

\paragraph{Monotonicity of Mechanism~\ref{alg:augm_greedy_fix}.} We  are now ready to prove that our mechanism is monotone. As a roadmap, we reduce an increase of the bid of agent~$i$ to a sequence of changes, each of which swaps an element of agent~$i$ with an element that is tied with it, either in the density ordering or in the value ordering, and show that no such change can move agent~$i$ to a lower position (Lemmas~\ref{lem:E1} and~\ref{lem:E2}). The case of the density ordering is simple, as only one of the sets of examined pairs changes. The case of the value ordering is the core of the proof: there, the run may indeed continue further than before, and we show that the additional pairs can only place agents below agent~$i$, so that her position cannot get worse (\Cref{thm:mono}).

\begin{theorem}\label{thm:mono}
Mechanism~\ref{alg:augm_greedy_fix} is monotone.
\end{theorem}

Consider profiles $(v_i,\valsmt)$ and $(v'_i,\valsmt)$ with $v'_i>v_i$, and let $x$ and $x'$
be the corresponding outputs of Mechanism~\ref{alg:augm_greedy_fix}. We show that if agent~$i$ is assigned to
position~$j$ in $x$, then in $x'$ she is assigned to some position $j'\le j$.

Observe that the output of a run of Mechanism~\ref{alg:augm_greedy_fix} is determined by the orderings
$(\sigma,\pi)$, the sizes and the capacity; the values enter only through the orderings.

\paragraph{Increasing the bid as a sequence of swaps.}
Let the bid of agent~$i$ increase continuously from $v_i$ to $v'_i$, and consider an
intermediate bid~$\beta$. The density and value orderings at profile $(\beta,\valsmt)$ change
only at finitely many values of~$\beta$, which we call \emph{critical}: those at which a
pair of agent~$i$ has the same density as a pair of another agent, or agent~$i$ has the
same value as another agent. (Two pairs of agent~$i$ never have the same density, since the
CTRs are strictly decreasing.) Between two consecutive critical values the orderings, and
hence the output of Mechanism~\ref{alg:augm_greedy_fix}, do not change. If $v_i$ or $v'_i$ is itself critical, the
orderings used at that bid are the ones of the mechanism at that bid, and the argument
below starts or ends there.

At a critical value~$\beta$, the orderings just before~$\beta$, at~$\beta$, and just
after~$\beta$ are all consistent with the profile $(\beta,\valsmt)$. Passing from one to the
next only moves a pair (or the position in the value ordering) of agent~$i$ ahead of a pair
(or an agent) that has the same density (or value) at~$\beta$. This is because:
\begin{itemize}
\item the relative order of the pairs of the other agents does not change, and neither does
the relative order of the pairs of agent~$i$;
\item the densities and the value of agent~$i$ only increase with her bid.
\end{itemize}

In particular, every two elements whose relative order differs between two consecutive
orderings have the same density (or value) at~$\beta$. As long as the current ordering
differs from the next one, some two elements that are adjacent in the current ordering
appear in the opposite order in the next one. Swapping them keeps the ordering consistent
with $(\beta,\valsmt)$, since they are tied at~$\beta$, and reduces the number of such
pairs by one. Therefore, we can pass from one ordering to the next by a finite sequence of
changes of the following two types, each of which keeps the orderings consistent with
$(\beta,\valsmt)$:
\begin{enumerate}[label=(E\arabic*)]
\item in $\sigma$, a pair $(i,j)$ moves ahead of the pair $(h,q)$ that immediately precedes
it, where $h\ne i$ and $d_{ij}=d_{hq}$;
\item in $\pi$, agent~$i$ moves ahead of the agent~$h$ that immediately precedes her, where
$v_i=v_h$.
\end{enumerate}
Theorem~\ref{thm:mono} then follows from the next two lemmas, applied to every such change
for every critical value~$\beta$ between $v_i$ and $v'_i$.

In both lemmas, we consider the runs of Mechanism~\ref{alg:augm_greedy_fix} with the orderings before and after the
change. We write $x(P_t)$ and $x'(P_t)$ for the allocations induced by the first $t$ pairs
in the two runs, and $T$ and $T'$ for the corresponding indices of
Corollary~\ref{cor:monotone-size}. We let $r$ be the position of agent~$i$ in $x(P_T)$; if
agent~$i$ is unassigned in $x(P_T)$, there is nothing to prove.

\begin{lemma}[Change (E1)]\label{lem:E1}
After a change of type (E1), agent~$i$ is assigned in $x'(P_{T'})$ to a position at most~$r$.
\end{lemma}

\begin{proof}
The sets $P_t$ before and after the change coincide, except for the set that contains
exactly one of the two swapped pairs. Let $P$ be the set of pairs preceding them.
\begin{itemize}
\item Before the change, the sets are $\dots,P,\;P\cup\{(h,q)\},\;P\cup\{(h,q),(i,j)\},\dots$
\item After the change, they are $\dots,P,\;P\cup\{(i,j)\},\;P\cup\{(h,q),(i,j)\},\dots$
\end{itemize}
Both orderings are consistent with $(\beta,\valsmt)$, so by Corollary~\ref{cor:monotone-size}
the total size of both middle allocations lies between $|x(P)|$ and
$|x(P\cup\{(h,q),(i,j)\})|$.

If $x(P)$ is infeasible, both runs stop at the same point, or earlier, and give the same
output. If $x(P\cup\{(h,q),(i,j)\})$ is feasible, both middle allocations are feasible, all
later sets coincide, and both runs again give the same output. It remains to consider the
case where $x(P)$ is feasible and $x(P\cup\{(h,q),(i,j)\})$ is not.

If the output before the change is $x(P)$, the output after the change is either $x(P)$ or
$x(P\cup\{(i,j)\})$. In the latter case, agent~$i$ is assigned to a position at most~$r$
by Lemma~\ref{lem:levels}.

If the output before the change is $x(P\cup\{(h,q)\})$, the output after the change is
either $x(P)$ or $x(P\cup\{(i,j)\})$. By Lemma~\ref{lem:levels}, the position of agent~$i$
in $x(P\cup\{(i,j)\})$ is at most her position in $x(P)$, which in turn is at most her
position in $x(P\cup\{(h,q)\})$.
\end{proof}

\begin{lemma}[Change (E2)]\label{lem:E2}
After a change of type (E2), agent~$i$ is assigned in $x'(P_{T'})$ to a position at most~$r$.
\end{lemma}

\begin{proof}
The ordering $\sigma$ does not change, so neither do the sets $P_t$. In $\pi$, agent~$h$
immediately precedes agent~$i$ before the change, and agent~$i$ immediately precedes
agent~$h$ after it, with $v_i=v_h$.

Fix a set $P=P_t$, let $c=c(P)$ be the counter of Lemma~\ref{lem:scan} just before agents
$h$ and $i$ are examined (it is the same before and after the change, as it depends only on
the agents that precede both), and write $\ell_h,\ell_i$ for their eligibility in~$P$.
Since the agents after $h$ and $i$ are scanned starting from the counter that $h$ and $i$
leave behind, Lemma~\ref{lem:scanfacts}(b) gives the following cases.
\begin{enumerate}[label=(\Alph*)]
\item $\ell_h\ge c+2$ and $\ell_i\ge c+2$. Both agents are assigned before and after the
change, so $x'(P)$ and $x(P)$ assign the same agents. Agent~$i$ moves from position $c+2$ to
position $c+1$.
\item $\ell_h=c+1$ and $\ell_i\ge c+2$. Agent~$h$ is assigned in $x(P)$ but not in $x'(P)$,
agent~$i$ is assigned to position $c+1$ in $x'(P)$, and at most one agent~$g$ that was
unassigned in $x(P)$ is assigned in $x'(P)$, to position $\ell_g(P)\ge c+2$. Since
$(g,\ell_g(P))\in P$ while $(h,c+2)\notin P$, we have $d_{g\ell_g(P)}\ge d_{h,c+2}$. Together
with $v_g\le v_h$ and $a_{\ell_g(P)}\le a_{c+2}$, this gives $s_g\le s_h$. Hence
$|x'(P)|\le|x(P)|$.
\item $\ell_i=c+1$ and $\ell_h\ge c+2$. Agent~$i$ is unassigned in $x(P)$; in $x'(P)$ she is
assigned to position $c+1$ and agent~$h$ to position $c+2$, and at most one agent~$g$ that
was assigned in $x(P)$, to position $\ell_g(P)\ge c+2$, becomes unassigned. Since
$(g,\ell_g(P))\in P$ while $(i,c+2)\notin P$, the same argument gives $s_g\le s_i$. Hence
$|x'(P)|\ge|x(P)|$.
\item $\ell_h=\ell_i=c+1$. Then $x'(P)$ is obtained from $x(P)$ by replacing agent~$h$ with
agent~$i$ in position $c+1$.
\item $\ell_h\le c$ and $\ell_i\ge c+1$. Then $x'(P)=x(P)$, and agent~$i$ is in position
$c+1$ in both.
\item $\ell_i\le c$. Agent~$i$ is unassigned before and after the change, and
$x'(P)=x(P)$.
\end{enumerate}
In all cases other than (B) and (D) we have $|x'(P)|\ge|x(P)|$, and in cases (B) and (D) we
have $\ell_h=c+1$. Moreover, if agent~$i$ is unassigned in $x'(P)$, we are in case (F) and
$x'(P)=x(P)$.

\emph{Step 1: $T'\ge T$.} Let $t\le T$. If agent~$i$ is unassigned in $x'(P_t)$, then
$x'(P_t)=x(P_t)$, which is feasible. Otherwise, $x'(P_t)=x'^{\,i}(P_t)$.

Since agent~$i$ is assigned in $x(P_T)$, at $P_T$ we are in case (A), (B) or (E). In each of
these, agent~$i$ is assigned in $x'(P_T)$ and $|x'(P_T)|\le|x(P_T)|\le W$. By
Lemma~\ref{lem:forced}, applied to the run after the change,
\[
|x'(P_t)|=|x'^{\,i}(P_t)|\le|x'^{\,i}(P_T)|=|x'(P_T)|\le W.
\]

\emph{Step 2.} Let $t>T$ and suppose that $x'(P_t)$ is feasible. By
Corollary~\ref{cor:monotone-size}, $x'(P_{t'})$ is feasible for every $t'\in(T,t]$, while
$x(P_{t'})$ is not. Hence every such $t'$ falls in case (B) or (D), so
$\ell_h(P_{t'})=c(P_{t'})+1$ for every $t'\in(T,t]$.

We claim that $c$ is constant on $(T,t]$. Consider two consecutive indices $t'-1,t'\in(T,t]$;
the set $P_{t'}$ is obtained from $P_{t'-1}$ by adding one pair, and both indices satisfy
$\ell_h=c+1$.
\begin{itemize}
\item If the added pair belongs to an agent that precedes $h$ and $i$ in $\pi$, then
$\ell_h$ does not change, so $c$ does not change either.
\item If it belongs to agent~$h$, then $\ell_h$ increases by one, so $c$ would have to
increase as well; this is impossible, since $c$ depends only on the agents that precede
$h$ and $i$.
\item Otherwise, neither $c$ nor $\ell_h$ changes.
\end{itemize}

At $P_T$, either agent~$h$ is assigned in $x(P_T)$, in which case $c(P_T)=r-2$, or she is
unassigned, in which case $c(P_T)=r-1$ and $\ell_h(P_T)\le r-1$.
\begin{itemize}
\item In the first case, $c(P_{T+1})\le c(P_T)+1=r-1$, since examining one pair increases
the counter by at most one (Lemma~\ref{lem:one-pair}(iv)).
\item In the second case, $\ell_h(P_{T+1})=c(P_{T+1})+1\ge c(P_T)+1=r>\ell_h(P_T)$, using that
the counter does not decrease (Lemma~\ref{lem:one-pair}(iv)). So the pair examined at step
$T+1$ belongs to agent~$h$, and $c(P_{T+1})=c(P_T)=r-1$.
\end{itemize}
In both cases $c(P_t)=c(P_{T+1})\le r-1$, and in cases (B) and (D) agent~$i$ is assigned to
position $c(P_t)+1\le r$ in $x'(P_t)$.

Combining the two steps, if $T'=T$ then agent~$i$ is assigned to a position at most~$r$
in $x'(P_T)$ by the case analysis at $P_T$, and if $T'>T$ then Step~2 applies with $t=T'$.
\end{proof}

\begin{proof}[Proof of Theorem~\ref{thm:mono}]
By Lemmas~\ref{lem:E1} and~\ref{lem:E2}, no change of type (E1) or (E2) can move agent~$i$ to a
lower position or leave her unassigned if she was assigned before. Applying them to every
change for every critical value~$\beta$ between $v_i$ and $v'_i$, agent~$i$ is assigned in
$x'$ to a position $j'\le j$. Therefore $a_{j'}\ge a_j$, that is, the click-through rate $q_i(b)=\sum_{j\in K}a_jx_{ij}(b)$ that agent~$i$ receives is non-decreasing in her bid (with $q_i(b)=0$ if she is unassigned). Hence, for her true value $v_i$, increasing her bid cannot result in her being assigned to a position of strictly lower realized value $v_i\cdot a_j$, and Mechanism~\ref{alg:augm_greedy_fix} is monotone.
\end{proof}

Combined with Myerson's payment scheme~\cite{Myerson}, this yields a truthful mechanism.

\subsection{A Universally Truthful Mechanism with Constant Approximation}\label{sec:mon}

As in Section~\ref{sec:alg}, Mechanism~\ref{alg:augm_greedy_fix} alone cannot guarantee a constant approximation. Consider, for
instance, a light agent with very high density and a heavy agent of much higher value whose
inclusion would leave no room for the former. We therefore randomize with $G(v_{max})$,
which outputs the single highest-value feasible pair, namely the agent with the highest value
among those with $s_i\le W$, assigned to position~$1$. Since the size of every agent is at
most $W$, $G(v_{max})$ is clearly monotone.

\begin{alg}\label{alg:final_mech_fix}
Consider a value profile $(v_1,v_2,\dots,v_n)$. Let $x_1$ be the output of Mechanism~\ref{alg:augm_greedy_fix} and
$x_2$ the output of $G(v_{max})$. Output: with probability $\frac34$ output $x_1$, and with
probability $\frac14$ output $x_2$.
\end{alg}

We use the notation of Section~\ref{sec:aprx_guar_main_resul}. We denote by $M^G$ the matching produced by
Mechanism~\ref{alg:augm_greedy_fix}, that is $M^G=x(P_T)$, and by $M^*$ an optimal matching with available
capacity~$W$. For any matching $M$, $v(M)=\sum_{(i,j)\in M}v_{ij}$, and $SW^R$ denotes the
expected social welfare of Mechanism~\ref{alg:final_mech_fix}.

We begin with the following lemma, which bounds the value of any matching that uses only
examined pairs, ignoring the capacity constraint.

\begin{lemma}\label{lem:charging}
Let $P=P_t$, and let $M$ be any matching such that every pair $(i,j)\in M$ belongs to $P$,
regardless of the capacity. Then $v(M)\le 2\cdot v(x(P))$.
\end{lemma}

\begin{proof}
Consider an agent~$i$ that is assigned to position~$j$ in $M$. Since $(i,j)\in P$, we have
$\ell_i(P)\ge j$. Let $c=c_i(P)$ be the counter of Lemma~\ref{lem:scan} just before agent~$i$
is examined. We distinguish two cases.
\begin{itemize}
\item If agent~$i$ is unassigned in $x(P)$, then $c\ge\ell_i(P)\ge j$. If agent~$i$ is
assigned in $x(P)$ to a position $c+1>j$, then again $c\ge j$. In both cases position~$j$ is
occupied in $x(P)$ by an agent who precedes agent~$i$ in the value ordering, and hence has
value at least $v_i$. We bound $v_{ij}$ by the value of the agent occupying position~$j$ in
$x(P)$. Since $M$ is a matching, every position of $x(P)$ is used in this way at most once.
\item Otherwise, agent~$i$ is assigned in $x(P)$ to a position $c+1\le j$, and
$v_{ij}=v_i\cdot a_j\le v_i\cdot a_{c+1}$. We bound $v_{ij}$ by the value of agent~$i$
herself in $x(P)$. Every agent of $x(P)$ is used in this way at most once.
\end{itemize}
Every agent of $x(P)$ is therefore used at most twice, and $v(M)\le 2\cdot v(x(P))$.
\end{proof}

\begin{theorem}\label{thm:tr-apx}
Mechanism~\ref{alg:final_mech_fix} is universally truthful, and provides (in expectation) a $4$-approximation of
the optimal social welfare, for the position auctions with a capacity constraint setting.
Moreover, this bound is tight.
\end{theorem}

\begin{proof}
Mechanism~\ref{alg:final_mech_fix} is universally truthful, because it is a randomization between two
deterministic monotone mechanisms: Mechanism~\ref{alg:augm_greedy_fix}, by Theorem~\ref{thm:mono}, and $G(v_{max})$.

For the approximation, we first show that
\[
v(M^*)\le 3\cdot v(M^G)+v(G(v_{max})).
\]
If Mechanism~\ref{alg:augm_greedy_fix} examines all pairs without stopping, then every pair of $M^*$ belongs to
$P_{n\cdot k}$, and Lemma~\ref{lem:charging} gives $v(M^*)\le 2\cdot v(M^G)$. Otherwise, let $(i^e,j^e)$ be the pair at which Mechanism~\ref{alg:augm_greedy_fix} stops, that is, the
$(T+1)$-th pair of the density ordering, and let $d^e=d_{i^ej^e}$. We split $M^*$ into two
parts and bound their contributions separately.
\begin{itemize}
    \item \emph{Pairs in $P_T$.} These pairs form a matching that uses only pairs of $P_T$, so by
Lemma~\ref{lem:charging} their contribution is at most $2\cdot v(x(P_T))=2\cdot v(M^G)$.
\item \emph{Pairs not in $P_T$.} Each such pair $(i,j)$ is either the pair $(i^e,j^e)$ or comes after it in the density
ordering, so $d_{ij}\le d^e$. Since $M^*$ respects the capacity, their contribution is
\[
\sum v_{ij}=\sum d_{ij}\cdot s_i\le d^e\cdot W.
\]
\end{itemize}
On the other hand, every pair used by $x(P_{T+1})$ belongs to $P_{T+1}$ and thus has density
at least $d^e$, while $|x(P_{T+1})|>W$ because the mechanism stops at $(i^e,j^e)$. Moreover, $v_{i^ej^e}\le v(G(v_{max}))$, since $s_{i^e}\le W$.
Therefore, using Lemma~\ref{lem:one-pair}(iii),
\[
d^e\cdot W<d^e\cdot|x(P_{T+1})|\le v(x(P_{T+1}))\le v(x(P_T))+v_{i^ej^e}\le v(M^G)+v(G(v_{max})).
\]

Putting everything together,
\begin{align*}
v(M^*) &\le 2\cdot v(M^G)+\bigl(v(M^G)+v(G(v_{max}))\bigr)\\
&= 3\cdot v(M^G)+v(G(v_{max}))\\
&= 4\cdot\Bigl(\tfrac34\cdot v(M^G)+\tfrac14\cdot v(G(v_{max}))\Bigr)\\
&= 4\cdot SW^R,
\end{align*}
where the last equation comes from the definition of the expected social welfare of
Mechanism~\ref{alg:final_mech_fix}. Notice that, given this bound, the probabilities $\frac34$ and $\frac14$ are the
best possible choice: if Mechanism~\ref{alg:augm_greedy_fix} is chosen with probability~$p$, the guarantee is
$\max\{3/p,\,1/(1-p)\}$, which is minimized for $p=\frac34$.

\paragraph{Lower Bound of Mechanism~\ref{alg:final_mech_fix}.}

Consider an instance with $2$ agents, $2$ positions with arbitrary CTRs $a_1>a_2>0$, and
capacity $W=1$. Agent~$1$ has $(v_1,s_1)=(\epsilon,\epsilon^2)$ and agent~$2$ has
$(v_2,s_2)=(1,1)$, for some arbitrarily small $\epsilon>0$ with $\epsilon<a_2/a_1$. Then the densities satisfy $d_{11}=a_1/\epsilon>d_{12}=a_2/\epsilon>d_{21}=a_1>d_{22}=a_2$, so the density ordering is $(1,1),(1,2),(2,1),(2,2)$.

Mechanism~\ref{alg:augm_greedy_fix} first assigns agent~$1$ to position~$1$, and examining $(1,2)$ does not change
the allocation. When $(2,1)$ is examined, agent~$2$ would take position~$1$ and agent~$1$
would move to position~$2$. The resulting allocation has size $1+\epsilon^2>W$, so the
mechanism stops, with social welfare $\epsilon\cdot a_1$. $G(v_{max})$ assigns agent~$2$ to
position~$1$, with social welfare $a_1$, which is also the optimal social welfare, since the
two agents do not fit together. The approximation ratio is then
\[
\frac{E[ALG]}{OPT}=\frac{\frac34\cdot\epsilon\cdot a_1+\frac14\cdot a_1}{a_1}
=\frac34\cdot\epsilon+\frac14\;\to\;\frac14\quad\text{as }\epsilon\to0.
\]
Therefore, the approximation ratio of Mechanism~\ref{alg:final_mech_fix} can be arbitrarily close to $4$, and our analysis is tight.
\end{proof}

\paragraph{On the two mechanisms.} We note that Mechanisms~\ref{alg:final_mech} and~\ref{alg:final_mech_fix} address different settings. Mechanism~\ref{alg:final_mech} applies to the general matching with a capacity constraint problem, where each agent may value the positions arbitrarily, and achieves a $6$-approximation, but it is not monotone (\Cref{prop:nonmon}), even in the position auctions setting. Mechanism~\ref{alg:final_mech_fix} is universally truthful and achieves a tight $4$-approximation, but its definition and analysis rely on the position auctions structure, where all agents share the same ranking over the positions: a displaced agent moves to the next position, agents are compared through a single value ordering, and eligibility for a position implies eligibility for all better positions. None of these is available when the values $v_{ij}$ are arbitrary, so Mechanism~\ref{alg:augm_greedy_fix} does not directly extend to the general problem. Whether a truthful mechanism with a constant approximation exists for the general problem, with single-parameter agents and ad-dependent CTRs, remains open (see Section \ref{sec:conclusion}).

\subsection{Experimental analysis}

\label{sec:simulations}
\begin{figure*}[t]
    \centering
    \begin{subfigure}[t]{0.48\textwidth}
        \centering
        \includegraphics[width=\linewidth]{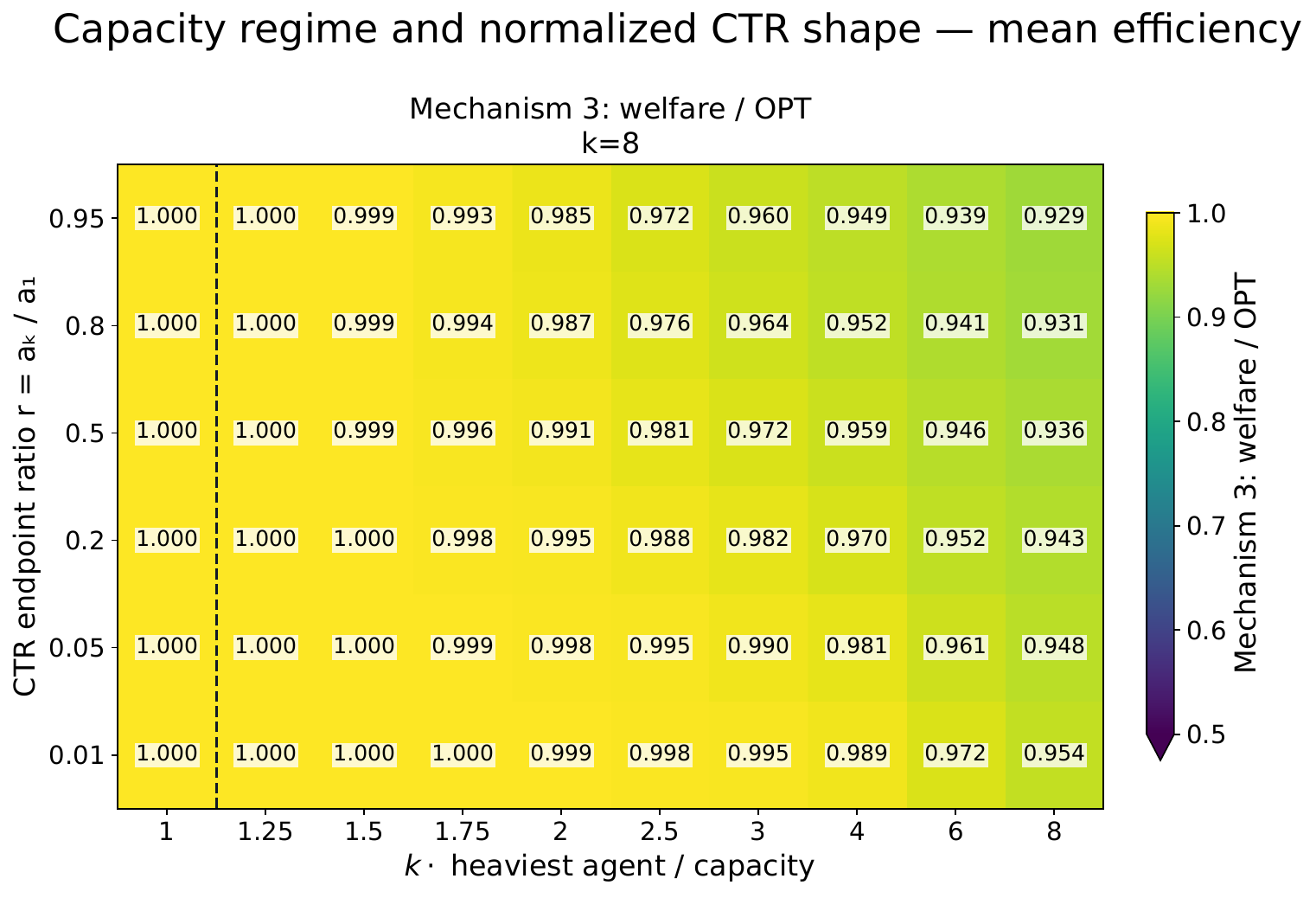}
    \end{subfigure}
    \hfill
    \begin{subfigure}[t]{0.48\textwidth}
        \centering
        \includegraphics[width=\linewidth]{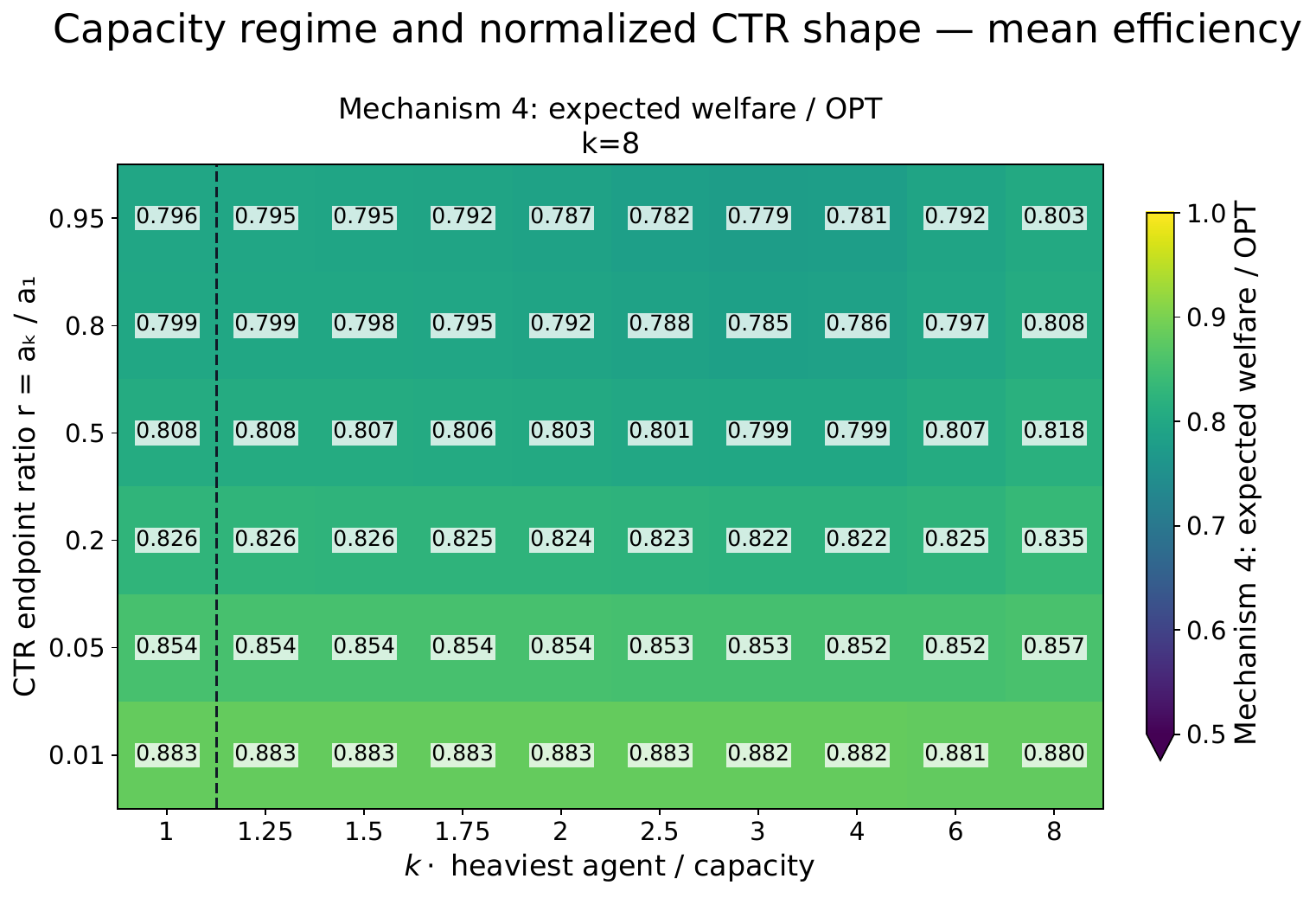}
    \end{subfigure}
    \caption{Average ratio between the social welfare of  Mechanism \ref{alg:augm_greedy_fix} (left) and Mechanism \ref{alg:final_mech_fix} (right) and the optimal (higher is better), as a function of $\capratio$ and of the CTR ratio $\dec=a_k/a_1$, with $n=12$ agents and $k=8$ positions. The values and sizes are drawn from  \(v_i \sim \mathrm{Uniform}(1,100)\), \(s_i \sim \mathrm{Uniform}(0.01,1)\) respectively.
    The dashed line marks $\capratio=1$.}
    \label{fig:heatmaps-ctr}
\end{figure*}

\paragraph{The effect of relative size ratio and the structure of CTR sequence.}
In this section, we complement our worst-case analysis with an average-case analysis of Mechanism \ref{alg:augm_greedy_fix} and Mechanism \ref{alg:final_mech_fix}, where the values and the sizes of the agents are drawn from fixed distributions, rather than chosen \textit{adversarially}. We identify two factors that affect the approximation ratio of our mechanisms: the size of the heaviest agent relative to the capacity, and the structure of the CTR sequence, which captures the level of concentration of value in the highest positions.
We measure the former through the parameter $\capratio$, where $s_{max}$ is the size of the heaviest agent. When $\capratio\le 1$, any $k$ agents fit in the capacity, so the capacity constraint has no effect: Mechanism \ref{alg:augm_greedy_fix} examines all pairs without stopping, and outputs an optimal allocation. When $\capratio>1$, the capacity $W$ and the number of positions $k$ interact, and Mechanism \ref{alg:augm_greedy_fix} may become suboptimal. For the latter, we model CTRs by the sequence $a_j = \dec^{(j-1)/(k-1)}$, where $0<\dec<1$ is the ratio $a_k/a_1$ between the last and the first CTR.
% \footnote{This normalization keeps the CTR profiles comparable across different values of $k$.}
Smaller values of $\dec$ concentrate value in the first positions, while values of $\dec$ close to $1$ make the contribution of each position roughly uniform. For each configuration of $(\capratio, \dec)$, we run 4000 simulations and calculate the average ratio between the social welfare of each mechanism and the optimal social welfare.
% TODO: specify the number of agents n, the distributions of values and sizes, and how OPT is computed (e.g., exhaustive search).

As shown in \Cref{fig:heatmaps-ctr}, Mechanism \ref{alg:augm_greedy_fix} is optimal when $\capratio\le 1$, and its efficiency decreases as $\capratio$ grows beyond this value. Moreover, Mechanism \ref{alg:augm_greedy_fix} outperforms Mechanism \ref{alg:final_mech_fix} in every configuration. This is expected, since with probability $\frac14$ Mechanism \ref{alg:final_mech_fix} outputs only the highest-value agent, which is necessary for the worst-case guarantee, but reduces the welfare on average. The decay rate of the CTR sequence has a stronger effect on Mechanism \ref{alg:final_mech_fix} than on Mechanism \ref{alg:augm_greedy_fix}: as the CTR sequence becomes steeper (smaller $\dec$), more of the optimal welfare comes from the first position, and thus $G(v_{max})$ contributes more to the welfare of Mechanism \ref{alg:final_mech_fix}.

%%%%%%%%%%%%%%%%%%%%%%%%%%%
\section{Conclusions}\label{sec:conclusion}
%%%%%%%%%%%%%%%%%%%%%%%%%%%

Motivated by applications in position auctions, in this work we studied mechanism design for matching problems, when there is also a single capacity constraint. As our results suggest, it is possible to have universally truthful mechanisms, that provide constant approximation guarantees to the social welfare. There are several interesting questions that remain open. In particular, as the mechanism that we design is randomized, it would be nice to explore whether it is possible to come up with deterministic truthful mechanisms that have similar guarantees, or even examine if it is within reach to get polynomial-time approximation schemes that are also monotone. Finally, it would be both interesting and challenging to go beyond the setting of CTRs that depend only on the positions, and explore the more general version of the problem, where agents continue to be single-parameter but the CTR of a position is affected by the ad that is assigned to it. This is the mechanism design equivalent to the algorithmic problem presented in Section \ref{sec:alg}, that is much more general than the one we study in this work. As we already mentioned, the monotone mechanism of Section~\ref{sec:mech} relies on the common ranking of the positions by all agents, and thus it cannot be directly applied to the general version of the problem. It is currently unclear if new approaches are needed, or a different modification to the stopping rule would be enough.

%%%%%%%%%%%%%%%%%%%%%%%%%%%%%%%%
\section*{Acknowledgments}
%%%%%%%%%%%%%%%%%%%%%%%%%%%%%%%%

Eleni Batziou was supported by EPSRC Grant EP/W014750/1.
\paragraph{Disclosure of AI Tools.} Claude Opus 5.5 was used for checking the correctness of the results, simplifying and improving the exposition in proofs of Section \ref{sec:mech}, and helping to develop the stopping rule present in Mechanism \ref{alg:augm_greedy_fix}. All results and proofs have been separately verified by the authors, who take full responsibility for them.

\bibliographystyle{plain}  
\bibliography{bibliography}

\end{document}